\documentclass{article}

\usepackage{setspace}
\usepackage[T1]{fontenc} 

\usepackage{enumitem}
\usepackage{setspace}
\usepackage{graphicx}
\usepackage{amsthm,amsmath,amssymb,enumitem} 
\usepackage{color}
\usepackage[framemethod=TikZ]{mdframed}
\usepackage[noadjust]{cite}
\usepackage{multirow}
\usepackage{tabu,array,soul}
\usepackage{subcaption}
\usepackage{caption}
\usepackage{algpseudocode}
\usepackage{siunitx}
\usepackage{textcomp}
\usepackage{tikz}
\usetikzlibrary{decorations.pathreplacing, calc}
\usetikzlibrary{positioning}
\usepackage{amsbsy}
\usepackage[final]{pdfpages}
\usepackage{complexity}
\usepackage{textcomp}
\usepackage{algorithmicx}
\usepackage[linesnumbered,ruled,vlined]{algorithm2e}

\newtheorem{observation}{Observation}
\newcommand{\overbar}[1]{\mkern 1.5mu\overline{\mkern-1.5mu#1\mkern-1.5mu}\mkern 1.5mu}

\mdfdefinestyle{MyFrame}{%
	linecolor=black,
	innerbottommargin=\baselineskip,
	backgroundcolor=gray!10!white}

\newcommand{\Inducedgraph}[3]{#1\left[#2,#3\right]}
\newcommand{\OPT}[2]{\Psi[#1,#2]}
\newcommand{\OPTT}{\Psi}
\newcommand{\Interval}[1]{I(#1)}
\newcommand{\leftend}[1]{l(#1)}
\newcommand{\rightend}[1]{r(#1)}

\newcommand{\Sep}[2]{S_{#1}^{#2}}
\newcommand{\minIndex}[1]{q^-_{#1}}
\newcommand{\maxIndex}[1]{q^+_{#1}}
\newcommand{\Level}[3]{L_{#1}\left(#2,#3\right)}
\newcommand{\dist}[3]{d_{#1}(#2,#3)}

\newcommand{\etal}{\textit{et al.}}

\newtheorem{theorem}{Theorem}
\newtheorem{definition}{Definition}

\newenvironment{remark}{\refstepcounter{observation}\medskip\noindent\textbf{Remark \theobservation. }}{\medskip}

\newtheorem{lemma}[theorem]{Lemma}
\newtheorem{question}{Question}

\newcommand{\Allseparate}[1]{\mathcal{S}\left(#1\right)}

\newtheorem*{myclaim}{Claim}
\newtheorem{claim}{Claim}
\newcommand{\helly}{\emph{Helly }}

\title{s-Club Cluster Vertex Deletion on Interval and Well-Partitioned Chordal Graphs\footnote{ An extended abstract~\cite{chakraborty2022s} of this paper was presented at WG 2022 and this version contains all proofs missing from the conference version. }} %TODO Please add

\author{Dibyayan Chakraborty\thanks{ENS de Lyon, France} \and L. Sunil Chandran\thanks{Indian Institute of Science, Bengaluru, India} \and Sajith Padinhatteeri \thanks{BITS-Pilani, Hyderabad.} \and Raji. R. Pillai \thanks{Indian Institute of Science, Bengaluru, India}}

\begin{document}

\maketitle

%TODO mandatory: add short abstract of the document
\begin{abstract}
In this paper, we study the computational complexity of \textsc{$s$-Club Cluster Vertex Deletion}. Given a graph, \textsc{$s$-Club Cluster Vertex Deletion ($s$-CVD)} aims to delete the minimum number of vertices from the graph so that each connected component of the resulting graph has a diameter at most $s$. When $s=1$, the corresponding problem is popularly known as \sloppy  \textsc{Cluster Vertex Deletion (CVD)}. We provide a faster algorithm for \textsc{$s$-CVD} on \emph{interval graphs}. For each $s\geq 1$, we give an $O(n(n+m))$-time algorithm for \textsc{$s$-CVD} on interval graphs with $n$ vertices and $m$ edges. In the case of  $s=1$, our algorithm is a slight improvement over the $O(n^3)$-time algorithm of Cao \etal (Theor. Comput. Sci., 2018) and for $s  \geq 2$, it significantly improves the state-of-the-art running time $\left(O\left(n^4\right)\right)$.

We also give a polynomial-time algorithm to solve \textsc{CVD} on \emph{well-partitioned chordal graphs}, a graph class introduced by Ahn \etal (\textsc{WG 2020}) as a tool for narrowing down complexity gaps for problems that are hard on chordal graphs, and easy on split graphs. Our algorithm relies on a characterisation of the optimal solution and on solving polynomially many instances of the \textsc{Weighted Bipartite Vertex Cover}. This generalises a result of Cao \etal (Theor. Comput. Sci., 2018) on split graphs.
 We also show that for any even integer $s\geq 2$, \textsc{$s$-CVD} is NP-hard on well-partitioned chordal graphs.   
 
\end{abstract}

\section{Introduction}

Detecting ``\emph{highly-connected}'' parts or ``\emph{clusters}'' of a complex system is a fundamental research topic in network science~\cite{yang2016comparative,papadopoulos2012community} with numerous applications in computational biology~\cite{dehne2006cluster,rahmann2007exact,ben1999clustering,sharan2000click,spirin2003protein}, machine learning~\cite{bansal2004correlation}, image processing~\cite{wu1993optimal}, etc. In a graph-theoretic approach, a complex system or a network is often viewed as an undirected graph $G$ that consists of a set of \emph{vertices} $V(G)$ representing the atomic entities of the system and a set of \emph{edges} $E(G)$ representing a binary relationship among the entities. A \emph{cluster} is often viewed as a dense subgraph (often a \emph{clique}) and \emph{partitioning} a graph into such clusters is one of the main objectives of \emph{graph-based data clustering}~\cite{ben1999clustering,shamir2004cluster,fellows2011graph}. 

Ben-Dor \etal~\cite{ben1999clustering} and Shamir \etal~\cite{shamir2004cluster} observed that the clusters of certain networks may be retrieved by making a small number of modifications in the network. These modifications may be required to account for the errors introduced during the construction of the network. In graph-theoretic terms, the objective is to modify (\textit{e.g.} edge deletion, edge addition, vertex deletion) a given input graph as little as possible so that each component of the resulting graph is a cluster. When deletion of vertices is the only valid operation on the input graph, the corresponding clustering problem falls in the category of \emph{vertex deletion} problems, a core topic in algorithmic graph theory. Many classic optimization problems like \textsc{Maximum Clique, Maximum Independent Set, Vertex cover} are examples of vertex deletion problems.  In this paper, we study popular vertex deletion problems called \textsc{Cluster Vertex Deletion} and its generalisation \textsc{$s$-Club Cluster Vertex Deletion}, both being important in the context of graph-based data clustering. 

Given a graph $G$, the objective of \textsc{Cluster Vertex Deletion (CVD)} is to delete a minimum number of vertices so that the  remaining graph is a set of disjoint cliques. Below we give a formal definition of \textsc{CVD}.

\bigskip 

\begin{mdframed}[style=MyFrame]
	\noindent \textsc{Cluster Vertex Deletion (CVD)}
	
	\noindent \textbf{Input:} An undirected graph $G$, and an integer $k$.
	
	 \noindent \textbf{Output:} \textsc{Yes}, if there is a set $S$ of vertices with $|S|\leq k$, such that each component of the graph induced by $V(G)\setminus S$ is a clique. \textsc{No}, otherwise.
\end{mdframed}

The term \textsc{Cluster Vertex Deletion} was coined by Gramm \etal~\cite{gramm2004automated} in $2004$. However NP-hardness of \textsc{CVD}, even on planar graphs and bipartite graphs, follows from the seminal works of Yannakakis~\cite{yannakakis1978} and Lewis \& Yannakakis~\cite{lewis1980} from four decades ago. Since then many researchers have proposed \emph{parameterized algorithms} and \emph{approximation algorithms} for \textsc{CVD} on general graphs~\cite{boral2016fast,tsur2021faster,huffner2010fixed,fomin2019exact,fomin2019subquadratic,sau2020hitting,you2017approximate,fiorini2016improved,fiorini2020improved,aprile2021tight}. In this paper, we focus on polynomial-time solvability of \textsc{CVD} on special classes of graphs. 

Cao \etal~\cite{cao2018vertex} gave polynomial-time algorithms for \textsc{CVD} on \emph{interval} graphs (see Definition~\ref{def:interval}) and \emph{split} graphs. Chakraborty \etal~\cite{chakraborty2021algorithms} gave a polynomial-time algorithm for \textsc{CVD} on \emph{trapezoid} graphs. However, much remains unknown: Chakraborty \etal~\cite{chakraborty2021algorithms} pointed out that computational complexity of \textsc{CVD} on \emph{planar bipartite} graphs and \emph{cocomparability} graphs is unknown. Cao \etal~\cite{cao2018vertex} asked if \textsc{CVD} can be solved on chordal graphs (graphs with no induced cycle of length greater than 3) in polynomial-time. 
% The above question remains open. 
% In this paper, we generalize a result of Cao \etal~\cite{cao2018vertex} and give the first polynomial-time algorithm for \textsc{CVD} on \emph{well-partitioned chordal} graphs~\cite{ahn2020well} (see Definition~\ref{def:well-partitioned} and Theorem~\ref{thm:well-partition-polytime}), a recently introduced subclass of chordal graphs that generalizes split graphs. 
Ahn \etal~\cite{ahn2020well} introduced \emph{well-partitioned chordal} graphs (see Definition~\ref{def:well-partitioned}) as a tool for narrowing down complexity gaps for problems that are hard on chordal graphs, and easy on split graphs. Since several problems (for example: transversal of longest paths and cycles, tree $3$-spanner problem, geodetic set problem) which are either hard or open on chordal graphs become polynomial-time solvable on well-partitioned chordal graphs~\cite{ahn2021three}, the computational complexity of \textsc{CVD} on well-partitioned chordal graphs is a well-motivated open question.

In this paper, we also study a generalisation of \textsc{CVD} known as \textsc{$s$-Club Cluster Vertex Deletion ($s$-CVD)}. In many applications the equivalence of cluster and clique is too restrictive~\cite{balasundaram2005novel,pasupuleti2008detection,alba1973graph}. For example, in protein networks where proteins are the vertices and the edges indicate the interaction between the proteins, a more appropriate notion of clusters may have a diameter of more than $1$~\cite{balasundaram2005novel}. Therefore researchers have defined the notion of \emph{$s$-clubs}~\cite{mokken,balasundaram2005novel}. An $s$-club is a graph with \emph{diameter} at most $s$. The objective of \textsc{$s$-Club Cluster Vertex Deletion ($s$-CVD)} is to delete the minimum number of vertices from the input graph so that all connected components of the resultant graph is an $s$-club. Below we give a formal definition of \textsc{$s$-CVD}.
\bigskip

\begin{mdframed}[style=MyFrame]
	\noindent \textsc{$s$-Club Cluster Vertex Deletion ($s$-CVD)}
	
	\noindent \textbf{Input:} An undirected graph $G$, and integers $k$ and $s$.
	
	\noindent \textbf{Output:} \textsc{Yes}, if there is a set $S$ of vertices with $|S|\leq k$, such that each component of the graph induced by $V(G)\setminus S$ has diameter at most $s$. \textsc{No}, otherwise.
\end{mdframed}

Sch{\"a}fer~\cite{schafer2009} introduced the notion of \textsc{$s$-CVD} and gave a polynomial-time algorithm for \textsc{$s$-CVD} on trees. Researchers have studied the particular case of \textsc{$2$-CVD} as well~\cite{liu2012,figiel2021}. In general, \textsc{$s$-CVD} remains NP-hard on planar bipartite graphs for each $s\geq 2$, APX-hard on split graphs for $s=2$ ~\cite{chakraborty2021algorithms} (contrasting the polynomial-time solvability of \textsc{CVD} on split graphs). 
% Recently, Chakraborty \etal~\cite{chakraborty2021algorithms} 
Combination of the ideas of Cao \etal~\cite{cao2018vertex} and Sch{\"a}fer~\cite{schafer2009}, provides an $O(n^8)$-time algorithm for \textsc{$s$-CVD} on a  trapezoid graphs (intersection graphs of trapezoids between two horizontal lines)  with $n$ vertices~\cite{chakraborty2021algorithms}. This algorithm can be modified to give an $O(n^4)$-time algorithm for \textsc{$s$-CVD} on interval graphs with $n$ vertices.

\medskip \noindent \textbf{General notations:}
For a graph $G$, let $V(G)$ and $E(G)$ denote the set of vertices and edges, respectively. For a vertex $v\in V(G)$, the set of vertices adjacent to $v$  is denoted by $N(v)$ and $N[v]=N(v)\cup \{v\}$. For $S \subseteq V(G)$, let $G-S$ be an induced graph obtained by deleting the vertices in $S$ from $G$. For two sets $S_1,S_2$, let $S_1-S_2$ denotes the set obtained by deleting the elements of $S_2$ from $S_1$. The set $S_1\Delta S_2$ denotes $(S_1\cup S_2)-(S_1\cap S_2)$. 

\medskip

% Notion of community, community detection, s-club

% CVD and s-CVD with applications, vertex deletion problems, 

% Chordal graphs, split graphs with application, well-partitioned chordal graphs, known results, our theorem

%interval graphs with applications, previous results on interval graphs, our result.

\section{Our Contributions}\label{sec:contribution}

In this section, we state our results formally. We start with the definition of well-partitioned chordal graphs as given in~\cite{ahn2020well}.

\begin{definition}[\cite{ahn2020well}]\label{def:well-partitioned}

A connected graph $G$ is a well-partitioned chordal graph if there exists a partition $\mathcal{P}$ of $V(G)$ and a
tree $\mathcal{T}$ having $\mathcal{P}$ as a vertex set such that the following hold.
\begin{enumerate}[label=(\alph*)]
    \item Each part $X\in \mathcal{P}$ is a clique in $G$.
    
    \item For each edge $XY\in E(\mathcal{T})$, there exist $X' \subseteq X$ and $Y' \subseteq Y$ such that edge set of the bipartite graph $G[X,Y]$ is $ X' \times Y'$.
    
    \item For each pair of distinct $X,Y\in V(\mathcal{T})$ with $XY\notin E(\mathcal{T})$, there is no edge between a vertex in $X$ and a vertex in $Y$.
\end{enumerate}

The tree $\mathcal{T}$ is called a \emph{partition tree} of $G$, and the elements of $\mathcal{P}$ are called its \emph{bags} or \emph{nodes} of $\mathcal{T}$.
\end{definition}

Our first result is on \textsc{CVD} for well-partitioned chordal graphs which generalises a result of Cao \etal~\cite{cao2018vertex}  for split graphs. We prove the following theorem in Section~\ref{sec:thm-well-partition}. 

% \begin{mdframed}[style=MyFrame]
% 	%\vspace{0.5em}
	\begin{theorem}\label{thm:well-partition-polytime}
    Given a well-partitioned chordal graph $G$ and its partition tree, there is an $O(m^2 n)$-time algorithm to solve \textsc{CVD} on $G$, where $n$ and $m$ are the number of vertices and edges.
	\end{theorem}
% 	\vspace{-0.5em}
% \end{mdframed}

Since a partition tree of a well-partitioned chordal graph can be obtained in polynomial time~\cite{ahn2020well}, the above theorem adds \textsc{CVD} to the list of problems that are open on chordal graphs but admits polynomial-time algorithm on well-partitioned chordal graphs.
% Hence, our algorithm is the first polynomial-time algorithm for \textsc{CVD} on well-partitioned graphs. 
Our algorithm relies on a characterisation of the solution set and we show that the optimal solution of a well-partitioned chordal graph with $m$ edges can be obtained by finding weighted minimum  vertex cover~\cite{kleinberg2006algorithm} of $m$ many weighted bipartite graphs with weights at most $n$. Then standard \emph{Max-flow} based algorithms~\cite{kleinberg2006algorithm,orlin2013max,king1994faster} from the literature yields Theorem~\ref{thm:well-partition-polytime}. On the negative side, we prove the following theorem in Section~\ref{sec:hardness}.

\begin{theorem}\label{thm:hardness}
    Unless the Unique Games Conjecture is false, for any even integer $s\geq 2$, there is no $(2-\epsilon)$-approximation algorithm for \textsc{$s$-CVD} on well-partitioned graphs.
\end{theorem}

Our third result is a faster algorithm for \textsc{$s$-CVD} on \emph{interval graphs}. 
\begin{definition}\label{def:interval}
A graph $G$ is an interval graph if there is a collection $\mathcal{I}$ of intervals on the real line such that each vertex of the graph can be mapped to an interval and two intervals intersect if and only if there is an edge between the corresponding vertices in $G$. The set $\mathcal{I}$ is an interval representation of $G$
\end{definition}

We prove the following theorem in Section~\ref{sec:thm-interval}.

% \begin{mdframed}[style=MyFrame]
% 	%\vspace{0.5em}
	\begin{theorem}\label{thm:interval-scd}
    For each $s\geq 1$, there is an $O(n(n+m))$-time algorithm to solve \textsc{$s$-CVD} on interval graphs with $n$ vertices and $m$ edges.
	\end{theorem}

We note that our techniques deviate significantly from the ones in the previous literature~\cite{schafer2009, cao2018vertex, chakraborty2021algorithms}. We show that the optimal solution (for \textsc{$s$-CVD} on interval graphs) must be one of ``four types'' and the optimum for each of the ``four types'' can be found by solving \textsc{$s$-CVD} on $O(m+n)$ many induced subgraphs. Furthermore, we exploit the ``linear'' structure of interval graphs to ensure that optimal solution in each case can be found in $O(n)$-time. Our result significantly improves the state-of-the-art running time $\left(O\left(n^4\right)\text{, See \cite{chakraborty2021algorithms}} \right)$ for \textsc{$s$-CVD} on interval graphs. 

% \medskip \noindent \textbf{Organisation:} In Sections~\ref{sec:hard}, \ref{sec:thm-well-partition} and \ref{sec:thm-interval}, we prove Theorem~\ref{thm:hard}, \ref{thm:well-partition-polytime} and \ref{thm:interval-scd}, respectively.

\section{Polynomial time algorithm for CVD on well-partitioned chordal graphs}\label{sec:thm-well-partition}

In this section, we shall give a polynomial-time algorithm to solve \textsc{CVD} on well-partitioned chordal graphs. In the next section, we present the main ideas of our algorithm and describe our techniques for proving Theorem~\ref{thm:well-partition-polytime}.

\subsection{Overview of the algorithm}

 Let $G$ be a well-partitioned chordal graph with a partition tree $\mathcal{T}$ rooted at an arbitrary node. For a node $X$, let $\mathcal{T}_X$ be the subtree rooted at $X$ and $G_X$ be the subgraph of $G$ induced by the vertices in the nodes of $\mathcal{T}_X$. For two adjacent nodes $X,Y$ of $\mathcal{T}$, \emph{the boundary of $X$ with respect to $Y$} is the set $bd(X,Y) = \{x\in X\colon N(x)\cap Y\neq \emptyset\}$. For a node $X$, $P(X)$ denotes the parent of $X$ in $\mathcal{T}$. We denote minimum CVD sets of $G_X$ and $G_X-bd(X,P(X))$ as $OPT(G_X)$ and $OPT(G_X-bd(X,P(X))$, respectively. We shall use the above notations extensively in the description of our algorithm and proofs. 
 
 Our dynamic programming-based algorithm traverses $\mathcal{T}$ in a post-order fashion and for each node $X$ of $\mathcal{T}$, computes $OPT(G_X)$ and $OPT(G_X-bd(X,P(X)))$. A set $S$ of vertices is a \emph{\textsc{CVD} set} of $G$ if $G-S$ is disjoint union of cliques. At the heart of our algorithm lies a characterisation of CVD sets of $G_X$, showing that any CVD set of $G_X$ can be exactly one of two types, namely, \emph{$X$-CVD set} or \emph{$(X,Y)$-CVD set} where $Y$ is a child of $X$ (See Definitions~\ref{def:CVD-type-(X)} and~\ref{def:CVD-type-(X,Y)}). Informally, for a node $X$, a CVD set is an $X$-CVD set if it contains $X$ or removing it from $G_X$ creates a cluster all of whose vertices are from $X$. On the contrary, a CVD set is an $(X,Y)$-CVD set if its removal creates a cluster intersecting both $X$ and $Y$, where $Y$ is a child of $X$. In Lemma~\ref{lem:CVD-characterisation}, we formally show that any CVD set of $G_X$ must be one of the above two types.
 
 To compute a minimum $X$-CVD set, first we construct a weighted bipartite graph $\mathcal{H}$ which is defined in Section~\ref{sec:(X)-CVD} and show that a minimum weighted vertex cover of $\mathcal{H}$ can be used to construct a minimum $X$-CVD set of $G$. (See Equations~\ref{eq:sol_1},~\ref{eq:sol_2},~\ref{eq:sol_3},~\ref{eq:sol}). Then in Section~\ref{sec:(X,Y)-CVD}, we show that the subroutine for finding minimum $X$-CVD sets can be used to to get a minimum $(X,Y)$-CVD set for each child $Y$ of $X$. Finally, in Section~\ref{sec:main-proof} we combine our tools and give an $O(m^2 n)$-time algorithm to find a minimum CVD set of an well-partitioned chordal graph $G$ with $n$ vertices and $m$ edges.

%  The objective of \textsc{cluster vertex deletion(CVD)} problem is to delete a minimum number of vertices of $G$ so that the remaining graph is a set of disjoint complete graphs. Below we present the intuition of our polynomial-time algorithm to find a minimum CVD set of a well-partitioned chordal graph $G$ with a partition tree $\mathcal{T}$. 

\subsection{Definitions and lemma}

In this section, we introduce some definitions and prove the lemma that facilitates the construction of a polynomial-time algorithm for finding a minimum CVD set of well-partitioned graphs. 

\begin{definition}
A \emph{cluster} $C$ of a graph $G$ is a connected component that is isomorphic to a complete graph.
\end{definition}

\begin{definition}\label{def:CVD-type-(X)}
Let $G$ be a well-partitioned graph, $\mathcal{T}$ be its partition tree, and $X$ be the root node of $\mathcal{T}$. A CVD set $S$ of $G$ is an $X$-CVD set if either $X\subseteq S$ or $G-S$ contains a cluster $C\subseteq X$.
\end{definition}

\begin{definition}\label{def:CVD-type-(X,Y)}
Let $G$ be a well-partitioned graph, $\mathcal{T}$ be its partition tree, $X$ be the root node of $\mathcal{T}$. Let $Y$ be a child of $X$.
%An edge $e\in E(G)$ is a an ``$(X,Y)$-edge'' if one of the endpoints of $e$ lies in $X$ and the other in $Y$.
A \textsc{CVD} set $S$ is a ``$(X,Y)$-CVD set'' if $G-S$ has a cluster $C$ such that $C \cap X \neq \emptyset$ and $C \cap Y \neq \emptyset$.
\end{definition}

\begin{lemma}\label{lem:CVD-characterisation}
Let $S$ be a CVD set of $G$. Then exactly one of the following holds.
\begin{enumerate}[label=(\alph*)]
    \item The set $S$ is a $X$-CVD set.
    \item There is exactly one child $Y$ of $X$ in $\mathcal{T}$ such that $S$ is an $(X,Y)$-CVD set of $G$.
\end{enumerate}
\end{lemma}

\begin{proof}
If $X\subseteq S$ or if $G-S$ has a cluster which is contained in $X$, then $S$ is an $X$-CVD set. Otherwise, $X^{*} =(G-S) \cap X \neq \emptyset$ and since $X^{*}$ is a clique, $G-S$ must contain a cluster $C$ such that $X^{*} \subset C \not \subseteq X$. Therefore, $C$ should intersect with at least one child of $X$. Let $Y_1, Y_2$ be children of $X$. If both $C \cap Y_1 \neq \emptyset$ and $C \cap Y_2 \neq \emptyset$, then $C$ is not a cluster because $Y_1$ and $Y_2$ are non-adjacent nodes of $\mathcal{T}$. Hence $C$ intersects exactly one child of $X$.
%Without loss of generality assume $a_1\in X$ and $b_1\in Y_1$. Now suppose  a child $Y_2$ of $X$ which is different from $Y_1$ such that $C$ contains an $(X,Y_2)$-edge $a_2b_2$ where $a_2\in X$ and $b_2\in Y_2$. Observe that $b_2$ and $b_1$ must be non-adjacent (otherwise there exists an edge between vertices of $Y_1$ and $Y_2$, which are non-adjacent nodes of $\mathcal{T}$; a contradiction). But then $b_1,b_2$ cannot be part of the same cluster $C$, a contradiction. Now suppose there are clusters $C_1,C_2$ such that for each $i\in \{1,2\}$, $C_i$ contains an $(X,Y'_i)$-edge $a'_ib'_i$ where $Y'_1\neq Y'_2$. Assume $a'_1,a'_2\in X$. Since $X$ induces a clique in $G$, using similar arguments as before, it is possible to show that $a'_1,a'_2,b'_1,b'_2$ induces a path of length at least $3$ in $G-S$, a contradiction. 
\end{proof}

\subsection{Finding minimum $X$-CVD sets}\label{sec:(X)-CVD}
%of well-partitioned chordal graphs
% For a node $X$ which is not the root of $\mathcal{T}$, let $P(X)$ denote the parent of $X$ in $\mathcal{T}$. For two adjacent nodes $X,Y$ of $\mathcal{T}$ let $bd(X,Y)$ denote the set of vertices of $X$ that has at least one neighbour in $Y$. For each node $X\in \mathcal{T}$, we shall denote the optimal \textsc{CVD} sets of $G_X$ and $G_X - bd(X,P(X))$ as $OPT(G_X)$ and $OPT(G_X - bd(X,P(X)))$.

In this section, we prove the following theorem.

\begin{theorem}\label{thm:private}
Let $G$ be a well-partitioned graph rooted at $X$ and $\mathcal{T}$ be a partition tree of $G$. Assume for each node $Y\in V(\mathcal{T})-\{X\}$ both $OPT(G_Y)$ and $OPT(G_Y-bd(Y,P(Y)))$ are given, where $P(Y)$ is the parent of $Y$ in $\mathcal{T}$. Then a minimum $X$-CVD set of $G$ can be computed in $O\left(|E(G)|.|V(G)|\right)$ time.
\end{theorem}

For the remainder of this section, we denote by $G$ a fixed well-partitioned graph rooted at $X$ with a partition tree $\mathcal{T}$. Let $X_1,X_2,\ldots,X_t$ be the children of $X$.  The main idea behind our algorithm for finding minimum $X$-CVD set of $G$ is to construct an auxiliary vertex weighted bipartite graph $\mathcal{H}$ with at most $|V(G)|$ vertices such that the (minimum) vertex covers of $\mathcal{H}$ can be used to construct (minimum) $X$-CVD-\textsc{CVD} set. Below we describe the construction of $\mathcal{H}$.

Let $\mathcal{B} = \left\{bd(X_i,X)\colon i\in [t]\right\}$. The vertex set of $\mathcal{H}$ is $X \cup \mathcal{B}$ and the edge set of $\mathcal{H}$ is defined as 
\begin{eqnarray}\label{eq:construct-1}
E(\mathcal{H}) &=& \{uB\colon u\in X, B\in \mathcal{B}, \forall v\in B, uv\in E(G)\}
 \end{eqnarray} 
The weight function on the vertices of $\mathcal{H}$ is defined as follows. For each vertex $u\in X$, define $w(u)=1$ and for each set $B\in \mathcal{B}$ where $B=bd(X_j,X)$, define 
\begin{eqnarray} 
\label{eq:weight}
w(B) &=& |B| + \left|OPT(G_{X_j} - B)\right| - \left|OPT (G_{X_j})\right|
\end{eqnarray} 

\begin{remark}
Since $B \cup OPT(G_{X_j}-B)$ is a CVD set of $G_{X_j}$, we have $ |OPT(G_{X_j})| \leq  |B| + |OPT(G_{X_j}-B)| $ and therefore $w(B)\geq 0$.
\end{remark}
Below we show how minimum weighted vertex covers of $\mathcal{H}$ can be used to compute minimum $X$-CVD set of $G$. For an $X$-CVD set $Z$ of $G$, define $Cov(Z)=(X\cap Z) \cup \left\{B \in \mathcal{B} \colon B \subseteq Z\right \}$.

\begin{lemma}\label{lem:sol-cover}
Let $Z$ be an $X$-CVD set of $G$. Then $Cov(Z)$ is a vertex cover of $\mathcal{H}$.
\end{lemma}
\begin{proof}
Assume that $Cov(Z)$ is not a vertex cover of $\mathcal{H}$. Then there exists at least one edge $e=uB$ in $\mathcal{H}- Cov(Z)$. Hence from the definition of $Cov(Z)$ we infer that $u \in X-Z$ and $B \not \subseteq Z$. Let $C_u$ be the cluster of $G-Z$ that contains the vertex $u$. Since $X$ is a clique, $X-Z \subseteq C_u$. Observe that since $uB$ is an edge of $\mathcal{H}$, there exists a vertex $w \in B$ such that $uw \in E(G)$. Then the definition of partition tree $\mathcal{T}$ and $B$ implies that all vertices of $B$ are contained in $N(u)$. Since  $B \not \subseteq Z$ it follows that there exists at least one vertex $v \in B$ in $G-Z$ such that $uv \in E(G-Z)$ and hence $v \in C_u$.  Therefore, the cluster $C_u$ intersects the child of $X$ that contains $B$ which contradicts the assumption that $Z$ is an $X$-CVD set of $G$ (see definition of $X$-CVD set). %and if $C$ is not a cluster then $Z$ cannot be a CVD set itself.
\end{proof}

For a vertex cover $D$ of $\mathcal{H}$, define

\begin{eqnarray}
    \label{eq:sol_1}
    S_1(D)&=&D\cap X \\
    \label{eq:sol_2}
    S_2(D)&=&  \displaystyle\bigcup\limits_{\substack{B\in D\cap \mathcal{B} \\ B=bd(X_i,X)}} B \cup OPT(G_{X_i} - bd(X_i,X)) \\ 
    \label{eq:sol_3}
    S_3(D) &=& \displaystyle\bigcup\limits_{\substack{B\in  \mathcal{B}- D \\ B=bd(X_i,X)}} OPT(G_{X_i}) \\
    \label{eq:sol}
    Sol(D) &=& S_1(D)\cup S_2(D) \cup S_3(D)
\end{eqnarray}

Note that, by definition $S_i(D) \cap S_j(D) = \emptyset, 1 \leq i < j \leq 3$. We have the following lemma.

\begin{lemma}
Let $D$ be a vertex cover of $\mathcal{H}$. Then $Sol(D)$ is an $X$-CVD set of $G$.
\end{lemma}
\begin{proof}
Suppose for the sake of contradiction that $Sol(D)$ is not an $X$-CVD set of $G$. First assume $Sol(D)$ is not a CVD set of $G$. Then there exists an induced path $P=uvw$ in $G-Sol(D)$. Consider the following cases.
\begin{enumerate}
    \item $X\cap \{u,v,w\}=\emptyset$. Then there must exist a child $Y$ of $X$ such that $u,v,w$ are vertices of $G_{Y}$. If $B=bd(Y,X)\in D$, then by Equations~\ref{eq:sol_2} and~\ref{eq:sol}, $Sol(D)$ contains $B\cup OPT(G_Y-B)$. But then $B\cup OPT(G_Y-B)$ is not a CVD set of $G_Y$, a contradiction. If $B\not\in D$, then by Equations~\ref{eq:sol_3} and~\ref{eq:sol}, $Sol(D)$ contains $OPT(G_Y)$.  But then $OPT(G_Y)$ is not a CVD set of $G_Y$, also a contradiction.
    
    \item Otherwise, there always exists two adjacent vertices $z_1,z_2$ such that $\{z_1,z_2\}\subset \{u,v,w\}$ and $z_1\in X$ and $z_2\in Y$, where $Y$ is a child of $X$. Observe that $z_2\in B=bd(Y,X)$ and therefore $z_1$ is adjacent to $B$ in $\mathcal{H}$. Since $\{z_1,z_2\} \cap Sol(D) = \emptyset$, $\mathcal{H}-D$ contains the edge $z_1 B$, contradicting the fact that $D$ is a vertex cover of $\mathcal{H}$.
\end{enumerate}
Now assume that $Sol(D)$ is a CVD set but not an $X$-CVD set. Then there must exists a cluster $C$ in $G-Sol(D)$ that contains an $(X,Y)$-edge $uv$ where $u\in X$ and $v\in bd(Y,X)$. Therefore $u\not\in D$ and $B=bd(Y,X)\not\in D$. Then $\mathcal{H}-D$ contains the edge $u B$, contradicting the fact that $D$ is a vertex cover of $\mathcal{H}$. 
\end{proof}

% It is easy to see that for $j\in [t]$ and any vertex cover $D$ of $\mathcal{H}^*$, either $bd(X,Y_j)\subseteq D$ or $bd(Y_j,X) \in  D$. From now on $D^*$ shall denote a minimum vertex cover of $\mathcal{H}^*$ and $Z$ shall denote a type-$0$ \textsc{CVD} set of $G^*$. We have the following observation.

A minimum weighted vertex cover $D$ of $\mathcal{H}$ is also \emph{minimal} if no proper subset of $D$ is a vertex cover of $\mathcal{H}$. The restriction of minimality is to avoid the inclusion of redundant vertices with weight $0$ in the minimum vertex cover.

\begin{observation}\label{obs:vertexcover}
Let $D$ be a minimal minimum weighted vertex cover of $\mathcal{H}$. For any $i\in [t]$, either $bd(X,X_i)\subseteq D$ or $bd(X_i,X) \in  D$, but not both.
\end{observation}

\begin{proof}
First assume $bd(X,X_i)\not\subseteq D$ and $B=bd(X_i,X)\not\in D$. Observe that, the neighbourhood of $B$ in $\mathcal{H}$ is $bd(X,X_i)$. Since $bd(X,X_i)\not\subseteq D$, there must exists a vertex $u\in (bd(x,X_i)-D) \subseteq X-D$. Then it follows that $uB$ is an edge of $\mathcal{H}-D$. This contradicts the fact that $D$ is a vertex cover of $\mathcal{H}$.\\
Now assume that both $bd(X,X_i)\subseteq D$ and $B=bd(X_i,X)\in D$. Since $\{x: xB \in E(\mathcal{H}\} = bd(X,X_i)$ the set $D-\{B\}$ is also a vertex cover of $\mathcal{H}$, a contradiction. 
\end{proof}
%If $w(B)=0$, then $D$ is not minimal
From now on $D$ denotes a minimal minimum weighted vertex cover of $\mathcal{H}$ and $Z$ denotes a fixed but arbitrary $X$-CVD set of $G$. Our goal is to show that $\left|Sol(D)\right|\leq \left|Z\right|$. We need some more notations and observations.

First we define four sets $I_1,I_2,I_3,I_4$ as follows. (Recall that $X_1,X_2,\ldots, X_t$ are children of the root $X$ of the partition tree $\mathcal{T}$ of $G$.) 
 \begin{eqnarray}
 \label{eq:I1}
     I_1 &=& \{i\in [t]\colon bd(X,X_i) \subseteq Sol(D) \text{ and }  bd(X, X_i) \subseteq Z  \} \\
     \label{eq:I2}
     I_2 &=& \{i\in [t]\colon bd(X,X_i) \subseteq Sol(D) \text{ and }   bd(X, X_i) \not\subseteq Z  \} \\
     \label{eq:I3}
    I_3 & = & \{i\in [t] - (I_1 \cup I_2)\colon bd(X_i,X) \subseteq Sol(D) \text{ and } bd(X_i,X) \subseteq Z \} \\
    \label{eq:I4}
     I_4 & = & \{i\in [t] - (I_1 \cup I_2)\colon bd(X_i,X) \subseteq Sol(D) \text{ and } bd(X_i,X) \not \subseteq Z \} 
 \end{eqnarray}
Note that $I_1 \cup I_2 \cup I_3 \cup I_4 =[t]$ and $(I_1 \cup I_2) \cap (I_3 \cup I_4) =\emptyset$. We have the following observations on the sets $I_i, 1 \leq i \leq 4$. 
  \begin{observation}\label{obs:partition}
 The sets $I_1,I_2,I_3,I_4$ form a partition of $[t]$.
 \end{observation}
   \begin{proof}

  From the definition of $I_i, 1 \leq i \leq 4$, it is clear that $I_i \cap I_j = \emptyset , i \neq j$. Assume that there exists an $i \in [t]$ such that $ i \notin I_1\cup I_2$. Hence, $bd(X,X_i) \not \subseteq  Sol(D) \cap X = D \cap X$.  Then by Observation \ref{obs:vertexcover}, $bd(X_i,X) \in D$ and by equation \ref{eq:sol_2}  the set of vertices $bd(X_i,X) \subseteq Sol(D)$. Therefore each $i \in [t] - (I_1\cup I_2)$ either belongs to the set $I_3$ or $I_4$.
  \end{proof}
\begin{observation}\label{obs:sets_I}
 Let $D$ be a vertex cover of $\mathcal{H}$ and $Sol(D)$ be an $X$-CVD set of $G$ defined as in equation \ref{eq:sol}. For the sets $I_i, 1 \leq i \leq 4$ defined by the Equations \ref{eq:I1} - \ref{eq:I4}, the following holds.\\

(i) $\bigcup\limits_{i \in I_1 \cup I_2} bd(X,X_i) = S_1(D)$ \\

 (ii) $\bigcup\limits_{i \in I_3 \cup I_4} bd(X_i,X) \cup OPT(G_{X_i} - bd(X_i,X) = S_2(D)$\\

 (iii) $\bigcup\limits_{i \in I_1 \cup I_2} OPT(G_{X_i}) = S_3(D)$.
 \end{observation}
 \begin{proof}
First note that $S_1(D) = D \cap X = Sol(D) \cap X$ (by definition of $Sol(D)$). On the other hand, by definition of $I_1$ and $I_2$ we have $\bigcup\limits_{i \in I_1 \cup I_2} bd(X,X_i) \subseteq Sol(D)$. Moreover, $\bigcup\limits_{i \in I_1 \cup I_2} bd(X,X_i) \subseteq X$. Therefore, $\bigcup\limits_{i \in I_1 \cup I_2} bd(X,X_i) \subseteq Sol(D)\cap X = S_1(D)$.\\
Now to prove the other side, $S_1(D) \subseteq \bigcup\limits_{i \in I_1 \cup I_2} bd(X,X_i)$,
%we have to show that for a vertex $v \in S_1(D)$ there exists an  $i \in I_1 \cup I_2$ such that $v \in bd(X,X_i)$.
suppose for the sake of contradiction that there exists a vertex $v \in S_1(D) - \bigcup\limits_{i \in I_1 \cup I_2} bd(X,X_i)$. Let $J = \{j : v \in bd(X,X_j)\}$. Since $J \cap (I_1 \cup I_2 ) =\emptyset$, by definition of $I_1$ and $I_2$, for each $j \in J, bd(X,X_j) \not \subseteq Sol(D) \cap X = D \cap X$. Hence by Observation \ref{obs:vertexcover}, $bd(X_j, X) \in D, \forall$. Therefore, $D-\{v\}$ is also a vertex cover of $\mathcal{H}$,  contradicting the minimality of $D$.\\

By Observation \ref{obs:partition}, $I_3 \cup I_4 = [t] - I_1 \cup I_2$. Moreover, by the definition of $I_1$ and $I_2$, for each $i \in [t] - I_1 \cup I_2$ the set $bd(X, X_i) \not\subseteq Sol(D) \cap X = D \cap X$. Hence by  Observation \ref{obs:vertexcover}, we have $bd(X_i, X) \in D$  for each $i \in I_3 \cup I_4$ and $bd(X_i, X) \notin D, i \in I_1 \cup I_2$. Thus it follows from Observation \ref{obs:partition} and the definition of $S_2$ and $S_3$ that $\bigcup\limits_{i \in I_3 \cup I_4} bd(X_i,X) \cup OPT(G_{X_i} - bd(X_i,X)) = S_2(D)$ and $\bigcup\limits_{i \in I_1 \cup I_2} OPT(G_{X_i}) = S_3(D)$.\\

% $\{bd(X_i, X): i \in I_3 \cup I_4\} = D \cap \mathcal{B}$

% Observe from the definition of $I_1$ and $I_2$ that for each $i \in [t]- (I_1 \cup I_2)$ the set $bd(X, X_i) \not\subseteq Sol(D) \not \subseteq D \cap X$ . Hence by  Observation \ref{obs:vertexcover},  we infer that $\{bd(X_i, X): i \in [t]- (I_1 \cup I_2)\} = D \cap \mathcal{B}$ and $\{bd(X_i, X): i \in I_1 \cup I_2\} = \mathcal{B}-D$. Since $\{bd(X_i, X): i \in [t]- (I_1 \cup I_2)\} = D \cap \mathcal{B}$, Equation \ref{eq:sol_2} implies $\{bd(X_i, X): i \in [t]- (I_1 \cup I_2)\} \subseteq Sol(D)$. Hence by the definition of $I_3$ and $I_4$ we infer that $[t]- (I_1 \cup I_2) = I_3 \cup I_4$. 
\end{proof}
Based on the set $I_1$, we construct two sets $D_1$ and $Z_1$ from $Sol(D)$ and $Z$, respectively, which are defined as follows.

\begin{eqnarray}
    %  D_1 &=& \bigcup\limits_{i \in I_1} bd(X,X_i) \cup OPT(G_{X_i}) \\ 
    D_1 &=& \bigcup\limits_{i \in I_1} bd(X,X_i) \cup (Sol(D)\cap G_{X_i})\\
     Z_1 &=&  \bigcup\limits_{i \in I_1} bd(X,X_i) \cup (Z \cap G_{X_i})
 \end{eqnarray} 

\begin{observation}\label{Obs:D1}
$  \left |D_1 \right| \leq \left |Z_1 \right |  $.
\end{observation}
 
 \begin{proof}
 From the definition of $Sol(D)$ and equation \ref{eq:sol_1}, for $i \in I_1$, we infer that $bd(X,X_i) \subseteq D$. Hence, by Observation \ref{obs:vertexcover},  $bd(X_i,X) \not \in D,  i \in I_1$ and from equation \ref{eq:sol_3}, $Sol(D) \cap G_{X_i} = OPT(G_{X_i})$. 
 Since for each $i, j \in I_1$, $G_{X_i} \cap G_{X_j} = \emptyset$ and $|Z \cap G_{X_i}| \geq |OPT(G_{X_i})|$, by the definitions of $D_1$ and $Z_1$ we have $|D_1| \leq |Z_1|$. 
 \end{proof}
 
Based on the set $I_2$, we construct the following two sets $D_2 \subseteq Sol(D)$ and $Z_2 \subseteq Z$.

\begin{eqnarray}
     D_2 &=& \bigcup\limits_{i \in I_2} bd(X,X_i) \cup (Sol(D)\cap G_{X_i})  - \bigcup\limits_{i \in I_1} bd(X,X_i)\\
     \label{eq:Z2}
     Z_2 &=& \bigcup\limits_{i \in I_2} bd(X_i,X) \cup  (Z \cap (G_{X_i} - bd(X_i,X)))     
 \end{eqnarray} 
 By the definition of the set $I_2$, the set of vertices $bd(X,X_i) \not\subseteq Z, i \in I_2$. By Lemma \ref{lem:sol-cover}, recall that there exits a vertex cover, $Cov(Z)$ of $\mathcal{H}$ corresponding to every $X$-CVD-set $Z$. Since $bd(X,X_i) \not\subseteq Z$ and thus  $bd(X,X_i) \not\subseteq Cov(Z)$, it is implicit in Observation \ref{obs:vertexcover} that $bd(X_i,X) \in Cov(Z)$. Hence $bd(X_i,X) \subseteq Z$ and the set $Z_2 \subseteq Z$.
 \begin{observation}\label{Obs:D2}
 $|D_2| \leq |Z_2|$.
 \end{observation}
  
 \begin{proof}
By arguments similar to that in the proof of Observation \ref{Obs:D1}, for $i \in I_2, (Sol(D)\cap G_{X_i}) = OPT(G_{X_i})$. Hence, $ D_2 = \bigcup\limits_{i \in I_2} bd(X,X_i) \cup OPT(G_{X_i})  - \bigcup\limits_{i \in I_1} bd(X,X_i)$.
 Suppose for contradiction that $|D_2| > |Z_2|$. Then by the definitions of $D_2$ and $Z_2$ we have
\begin{multline*}
   \left |\bigcup\limits_{i \in I_2}bd(X,X_i) \cup OPT(G_{X_i})  - \bigcup\limits_{i \in I_1} bd(X,X_i)\right | >  \left |\bigcup\limits_{i \in I_2}bd(X_i,X) \cup  (Z \cap (G_{X_i} - bd(X_i,X))) \right  |
    \end{multline*}
    Since $X \cap G_{X_i} = \emptyset, 1 \leq i \leq t$ and $\left |Z \cap (G_{X_i} - bd(X_i,X)) \right | \geq OPT(G_{X_i} - bd(X_i,X))$, we can rewrite the above inequality as follows.
%     \begin{multline*}
%     \left |\bigcup\limits_{i \in I_2}bd(X,X_i)  - \bigcup\limits_{i \in I_1} bd(X,X_i) \right | + \left |\bigcup\limits_{i \in I_2} OPT(G_{X_i}) \right | > \left |\bigcup\limits_{i \in I_2}bd(X_i,X) \right | + \\ \left |\bigcup\limits_{i \in I_2} OPT(G_{X_i} - bd(X_i,X)) \right |
%     \end{multline*}
% Rearranging the above inequality we have,
\begin{multline*}
   \left |\bigcup\limits_{i \in I_2}bd(X,X_i)  - \bigcup\limits_{i \in I_1} bd(X,X_i) \right | > \left |\bigcup\limits_{i \in I_2}bd(X_i,X) \right | + \\ \left |\bigcup\limits_{i \in I_2}OPT(G_{X_i} - bd(X_i,X)) \right | - \left |\bigcup\limits_{i \in I_2}OPT(G_{X_i}) \right |
    \end{multline*}
% Recall that by the definition of $bd(X_i,X)$, we know that $bd(X_i,X) \cap bd(X_j,X) = \emptyset$ for $i \neq j$. Hence $\left |\bigcup\limits_{i \in I_2}bd(X_i,X) \right | = \sum\limits_{i \in I_2}\left |bd(X_i,X) \right |$. Similarly, 
%  $G_{X_i} \cap G_{X_j} = \emptyset$ for $i \neq j$ and hence $\left |\bigcup\limits_{i \in I_2}OPT(G_{X_i} - bd(X_i,X)) \right | = \sum\limits_{i \in I_2} \left |OPT(G_{X_i} - bd(X_i,X)) \right |$ and $\left |\bigcup\limits_{i \in I_2}OPT(G_{X_i}) \right | = \sum\limits_{i \in I_2}\left |OPT(G_{X_i} ) \right |$. Therefore, we have
%  \begin{multline*}
%     \left |\bigcup\limits_{i \in I_2}bd(X,X_i)  - \bigcup\limits_{i \in I_1} bd(X,X_i) \right | > \sum\limits_{i \in I_2}\left |bd(X_i,X)\right | + \\ \sum\limits_{i \in I_2}\left |OPT(G_{X_i} - bd(X_i,X)) \right | - \sum\limits_{i \in I_2} \left |OPT(G_{X_i} ) \right |
%     \end{multline*}
  That is,
  \begin{multline*}
    \left |\bigcup\limits_{i \in I_2}bd(X,X_i)  - \bigcup\limits_{i \in I_1} bd(X,X_i) \right | > \sum\limits_{i \in I_2}(\left |bd(X_i,X)\right | + \left |OPT(G_{X_i} - bd(X_i,X) \right | - \left |OPT(G_{X_i} ) \right |)
    \end{multline*}
    
%      \begin{multline*}
%   |\bigcup\limits_{i \in I_2}bd(X,X_i)  - \bigcup\limits_{j \in I_1} bd(X,X_j)| > \sum\limits_{i \in I_2}(|bd(X_i,X)| + |OPT(G_{X_i} - bd(X_i,X))| - |OPT(G_{X_i} )|)
%     \end{multline*}
    
    By equation \ref{eq:weight}, $\left |bd(X_i,X) \right | + \left |OPT(G_{X_i} - bd(X_i,X)) \right | - \left |OPT(G_{X_i} ) \right | = w(bd(X_i,X)) $  and hence, \\
  
    \begin{align}\label{minimum_vertex}
       \left |\bigcup\limits_{i \in I_2}bd(X,X_i)  - \bigcup\limits_{i \in I_1} bd(X,X_i) \right | > \sum\limits_{i \in I_2}w(bd(X_i,X))
  \end{align}
Recall that $D$ is a minimal minimum weighted vertex cover of $\mathcal{H}$. By Observation \ref{obs:sets_I} we have $\bigcup\limits_{i \in I_2}bd(X,X_i) \subseteq D$ and hence for each $i \in I_2$, the vertex $B=bd(X_i,X) \notin D$ by Observation \ref{obs:vertexcover}. Now we show that if we delete the vertices in $\bigcup\limits_{i \in I_2}bd(X,X_i)  - \bigcup\limits_{i \in I_1} bd(X,X_i)$ from $D$ and add the set of vertices $\left \{ bd(X_i,X) \colon i \in I_2\right \}$ then we get a vertex cover of smaller weight for $\mathcal{H}$ by inequality (\ref{minimum_vertex}), a contradiction.  
\begin{claim}
Let $D_1$ be a set of vertices obtained from $D$ by deleting the vertices in $\bigcup\limits_{i \in I_2}bd(X,X_i)  - \bigcup\limits_{i \in I_1} bd(X,X_i)$ and by adding the set of vertices $\left \{ bd(X_i,X) \colon i \in I_2\right \}$. Then, $D_1$ is a vertex cover of $\mathcal{H}$.
\end{claim}
\begin{proof}[Proof of claim]
\renewcommand{\qedsymbol}{}
Assume that there exists an edge $uB \in E(\mathcal{H} - D_1)$ where $B = bd(X_j,X), j \in [t]$. Since $bd(X_j,X) \notin D_1$, by the definition of $D_1$ (given above ) observe that $bd(X_j,X) \notin D$ and $j \notin I_2$. Note that the neighbourhood of $bd(X_j,X)$ in $\mathcal{H}$  is $bd(X,X_j)$ and hence $u \in bd(X,X_j)$. Since $D$ is a vertex cover of $\mathcal{H}$, we have $bd(X,X_j) \subseteq D$.  Now we show that $j \notin I_1$: By definition of $D_1$ we have  $\bigcup\limits_{i \in I_1} bd(X,X_i) \cap D \subseteq D_1$. Since $u \in bd(X,X_j)$ and $bd(X,X_j) \subseteq D$, if $j \in I_1$ then the vertex $u$ remains in $D_1$. Thus no such edge $uB$ exists in $\mathcal{H} - D_1$. Therefore, we infer that  $j \notin I_1$. Since $j \notin I_1 \cup I_2$, from Observation \ref{obs:sets_I} we have $bd(X,X_j) \not \subseteq D \cap X$. Hence  there exists a vertex $w \in bd(X,X_j)$ such that $w \in \mathcal{H} - D$. Moreover, by the definition of partition tree $\mathcal{T}$ and $bd(X,X_j)$ the edge $wB \in E(\mathcal{H} - D)$. This contradicts the assumption that $D$ is a vertex cover of $\mathcal{H}$.  
\end{proof}
This completes the proof of the observation. 
\end{proof}
 
 Based on the set $I_3$, we construct the following two sets $D_3 \subseteq Sol(D)$ and $Z_3 \subseteq Z$.
 
 \begin{eqnarray}
      D_3 & = & \bigcup\limits_{i \in I_3} bd(X_i,X) \cup  OPT(G_{X_i} - bd(X_i,X)) \\
    Z_3 & = & \bigcup\limits_{i \in I_3}bd(X_i,X) \cup  (Z \cap (G_{X_i} - bd(X_i,X)))
 \end{eqnarray} 
 
  \begin{observation}
 $\left |D_3 \right | \leq \left |Z_3 \right |$.
 \end{observation}
 
 \begin{proof}
 Since $\left |Z \cap (G_{X_i}- bd(X_i,X)) \right | \geq OPT(G_{X_i} - bd(X_i,X))$, by the definitions of $D_3$ and $Z_3$ we have $\left |D_3 \right | \leq \left |Z_3 \right |$.
 \end{proof}
 
 Based on the set $I_4$, we construct the following two sets $D_4 \subseteq Sol(D)$ and $Z_4 \subseteq Z$.
 
 \begin{eqnarray}
      D_4 & = & \bigcup\limits_{i \in I_4} bd(X_i,X) \cup  OPT(G_{X_i} - bd(X_i,X))  \\
      \label{eq:Z4}
    Z_4 & = & \bigcup\limits_{i \in I_4} bd(X,X_i) \cup (Z \cap (G_{X_i})) - \bigcup\limits_{i \in I_1} bd(X,X_i)
 \end{eqnarray} 
 By the definition of the set $I_4$, the set of vertices $bd(X_i,X) \not\subseteq Z, i \in I_4$. By Lemma \ref{lem:sol-cover}, recall that there exits a vertex cover, $Cov(Z)$ of $\mathcal{H}$ corresponding to every $X$-CVD-set $Z$. Since $bd(X_i,X) \not\subseteq Z,  i \in I_4$, by definition of $Cov(Z)$ we have $bd(X_i,X) \notin Cov(Z)$ and hence it is implicit in Observation \ref{obs:vertexcover} that $bd(X,X_i) \subseteq Cov(Z)$. Hence $bd(X,X_i) \subseteq Z$ and the set $Z_4 \subseteq Z$.
  \begin{observation}
 $\left |D_4 \right | \leq \left |Z_4 \right |$.
 \end{observation}
 
 \begin{proof}
 Suppose for contradiction that $\left |D_4 \right | > \left |Z_4 \right |$. Then by the definitions of $D_4$ and $Z_4$ we have
\begin{multline*}
    \left |\bigcup\limits_{i \in I_4}(bd(X_i,X) \cup OPT(G_{X_i}- bd(X_i,X))) \right | > \left |\bigcup\limits_{i \in I_4}(bd(X,X_i) \cup  (Z \cap (G_{X_i})  - \bigcup\limits_{i \in I_1} bd(X,X_i) \right |
    \end{multline*}
 Since $G_{X_i} \cap G_{X_j} = \emptyset$ for $i, j \in I_4$ and $Z \cap (G_{X_i}) \geq OPT(G_{X_i})$, we have
%  \begin{multline*}
%   \left |\bigcup\limits_{i \in I_4}(bd(X_i,X) \right | + \left |\bigcup\limits_{i \in I_4}OPT(G_{X_i}- bd(X_i,X))) \right | > \left |\bigcup\limits_{i \in I_4}(bd(X,X_i) - \bigcup\limits_{i \in I_1} bd(X,X_i) \right | \\ +  \left |\bigcup\limits_{i \in I_4}OPT(G_{X_i}) \right |
%     \end{multline*}
%     Rearranging the inequality we get,
     \begin{multline*}
   \left |\bigcup\limits_{i \in I_4}(bd(X_i,X) \right | + \left |\bigcup\limits_{i \in I_4}OPT(G_{X_i}- bd(X_i,X))) \right | -  \left |\bigcup\limits_{i \in I_4}OPT(G_{X_i}) \right | \\ > \left |\bigcup\limits_{i \in I_4}(bd(X,X_i)   - \bigcup\limits_{i \in I_1} bd(X,X_i) \right | 
    \end{multline*}
    % Since  $bd(X_i,X) \cap bd(X_j,X) = \emptyset$ for $i, j \in I_4$, rewrite the above inequality as follows. 
    %  \begin{multline*}
    % \sum\limits_{i \in I_4}\left |(bd(X_i,X) \right | + \left |OPT(G_{X_i}- bd(X_i,X))) \right | -  \left |OPT(G_{X_i}) \right | \\ > \left |\bigcup\limits_{i \in I_4}(bd(X,X_i)  - \bigcup\limits_{i \in I_1} bd(X,X_i) \right | 
    % \end{multline*}
     Note that  by equation \ref{eq:weight}, $\left |bd(X_i,X) \right | + \left |OPT(G_{X_i} - bd(X_i,X)) \right | - \left |OPT(G_{X_i} ) \right | = w(bd(X_i,X)) $  and hence, 
     \begin{multline}\label{ineq:weight2}
         \sum\limits_{i \in I_4} w(bd(X_i,X)) > \left |\bigcup\limits_{i \in I_4}(bd(X,X_i) - \bigcup\limits_{i \in I_1} bd(X,X_i) \right |
     \end{multline}
     
     Recall that $D$ is a minimal minimum weighted vertex cover of $\mathcal{H}$. Observe that by definition of $I_1$ and $Sol(D)$, the set $\bigcup\limits_{i \in I_1} bd(X,X_i) \subseteq D$. Now we show that if we delete the vertices in $\left \{ bd(X_i,X)\colon i \in I_4 \right \}$ from $D$ and adding the set of vertices  $\bigcup\limits_{i \in I_4}bd(X,X_i)  - \bigcup\limits_{i \in I_1}bd(X,X_i)$, then we get a vertex cover of smaller weight for $\mathcal{H}$ by inequality (\ref{ineq:weight2}), a contradiction:   By definition of $I_1$ and $Sol(D)$, the set $\bigcup\limits_{i \in I_1} bd(X,X_i) \subseteq D$. Hence by the addition of the  vertices  $\bigcup\limits_{i \in I_4}bd(X,X_i)  - \bigcup\limits_{i \in I_1}bd(X,X_i)$ to $D$ we have the neighbourhood of each deleted vertex $bd(X_i,X)$ in $D$. 
    %  Let $D_1$ be a set of vertices obtained from $D$ by deleting the set of vertices $\left \{ bd(X_i,X)\colon i \in I_4 \right \}$ from $D$ and adding the set of vertices  $\bigcup\limits_{i \in I_4}bd(X,X_i)  - \bigcup\limits_{i \in I_1}bd(X,X_i)$. Observe that by definition of $I_1$ and $Sol(D)$, the set $\bigcup\limits_{i \in I_1} bd(X,X_i) \subseteq D$. Hence, for each deleted vertex $bd(X_i,X)\colon i \in I_4$ in $D$ the neighbourhood of  $bd(X_i,X)$ is included in $D_1$. Thus it follows that $D_1$ is also a vertex cover of $H$  and hence inequality (\ref{ineq:weight2}) contradicts the minimality of $D$.
 \end{proof}
 
 \begin{lemma}\label{lem:comparison}
 $Sol(D) = \bigsqcup\limits_{i=1}^{4} D_i$ and for each $i,j\subset [4]$, $Z_i\cap Z_j=\emptyset$.
 \end{lemma}
 
 \begin{proof}
   By Observation \ref{obs:partition} it follows from the definition that for $1 \leq i \neq j \leq 4$, the sets $D
   _i \cap D_j =\emptyset$ and $Z
   _i \cap Z_j =\emptyset$. \\
   
   Now we show that $Sol(D) = D_1 \cup D_2 \cup D_3 \cup D_4$. First consider the set $D_1 \cup D_2 = \bigcup\limits_{i \in I_1 \cup I_2} (bd(X,X_i) \cup (Sol(D)\cap G_{X_i}))$. By Observation \ref{obs:sets_I}, $\bigcup\limits_{i \in I_1 \cup I_2} bd(X,X_i)= S_1(D)$ and $\bigcup\limits_{i \in I_1 \cup I_2}(Sol(D)\cap G_{X_i}) = S_3(D)$. Hence $D_1 \cup D_2 = S_1(D) \cup S_3(D)$\\
   Now consider the set $D_3 \cup D_4 = \bigcup\limits_{i \in I_3 \cup I_4} bd(X_i,X) \cup OPT(G_{X_i}- bd(X_i,X))$. Hence by Observation \ref{obs:sets_I}, $D_3 \cup D_4 = S_2(D)$. Therefore, the definition of $Sol(D)$ (Equation \ref{eq:sol}) implies $Sol(D) = \bigsqcup\limits_{i=1}^{4} D_i$.
 \end{proof}
 
 \medskip \noindent \textbf{Proof of Theorem~\ref{thm:private}} Using Lemma~\ref{lem:comparison}, we have that $\left|Sol(D)\right|\leq |Z_1\cup Z_2\cup Z_3\cup Z_4|\leq |Z|$. Hence, $Sol(D)$ is a minimum $X$-CVD set of $G$. Furthermore, $\mathcal{H}$ has at most $|V(G)|$ vertices and $|E(G)|$ edges. Therefore minimum weighted vertex cover of $\mathcal{H}$ can be found in $O(|V(G)|\cdot|E(G)|)$-time and $Sol(D)$ can be computed in total of $O(|V(G)|\cdot|E(G)|)$-time. Below we give a short pseudocode of our algorithm to find a minimum $X$-CVD set of $G$.
 
\medskip
 \begin{minipage}{\linewidth}
 \begin{algorithm}[H] 
 \DontPrintSemicolon
\begin{small}
   \SetKwFunction{compute}{Compute\_sCD(G,a,A)}
\SetKwInOut{KwIn}{Input}
    \SetKwInOut{KwOut}{Output}
    \KwIn{A well-partitioned chordal graph $G$, a partition tree $\mathcal{T}$ of $G$ rooted at the node $X$, for each node $Y\in \mathcal{T}-\{X\}$ both $OPT(G_Y)$ and $OPT(G_Y-bd(Y,P(Y)))$ are given as part of input}
    \KwOut{A minimum $X$-CVD set}
    
    Construct a weighted bipartite graph $\mathcal{H}$ as described in Equations~\ref{eq:construct-1} and Equations~\ref{eq:weight};
    
    Find a minimum weighted vertex cover $D$ of $\mathcal{H}$;
    
    Construct the sets $S_1(D), S_2(D), S_3(D)$ and $Sol(D)$ as described in Equations~\ref{eq:sol_1},~\ref{eq:sol_2},~\ref{eq:sol_3} and~\ref{eq:sol}, respectively;
    
 \Return $Sol(D)$
   \caption{Pseudocode to find a minimum $X$-CVD set of a well-partitioned chordal graph}
   \label{alg:(X)-CVD}
    
\end{small}
 \end{algorithm}
\end{minipage}

\subsection{Finding minimum $(X,Y)$-CVD set of well-partitioned chordal graphs}\label{sec:(X,Y)-CVD}

In this section, we prove the following theorem.

\begin{theorem}\label{thm:(X,Y)}
Let $G$ be a well-partitioned graph; $\mathcal{T}$ be a partition tree of $G$ rooted at $X$; $Y$ be a child of $X$. Moreover, for each $Z\in V(\mathcal{T})-\{X\}$, assume both $OPT(G_Z)$ and $OPT(G_{Z}-bd(Z,P(Z)))$ are given $(P(Z)$ denotes the parent of $Z$ in $\mathcal{T})$. Then a minimum $(X,Y)$-CVD set of $G$ can be computed in $O\left(|E(G)|^2.|V(G)|\right)$ time.
\end{theorem}

For the remainder of this section, the meaning of $G$, $\mathcal{T}$, $X$ and $Y$ will be as given in Theorem~\ref{thm:(X,Y)}.  For an $(X,Y)$-edge $e$, we say that a minimum $(X,Y)$-CVD set $A$ ``preserves" the edge $e$ if $G-A$ contains the edge $e$. Let $e \in E(X,Y)$ be an $(X,Y)$-edges of $G$. Then to prove Theorem~\ref{thm:(X,Y)}, we use Theorem~\ref{thm:(X,Y)-edge}. First we show how to construct a minimum $(X,Y)$-CVD set $S_e$ that preserves the edge $e \in E(X,Y)$ and prove Theorem~\ref{thm:(X,Y)-edge}. Clearly, a minimum $(X,Y)$-CVD set $S$ of $G$ is the one that satisfies $|S|=\min\limits_{e \in E(X,Y)}|S_e|$. Therefore, Theorem~\ref{thm:(X,Y)} will follow directly from Theorem~\ref{thm:(X,Y)-edge}.  The remainder of this section is devoted to prove Theorem~\ref{thm:(X,Y)-edge}.

\begin{theorem}\label{thm:(X,Y)-edge}
Assuming the same conditions as in Theorem~\ref{thm:(X,Y)}, for $e \in E(X,Y)$, a minimum $(X,Y)$-CVD set of $G$ that preserves $e$ can be computed in $O\left(|E(G)|.|V(G)|\right)$ time.
\end{theorem}

First, we need the following observation about the partition trees of well-partitioned chordal graphs, which is easy to verify.

\begin{observation}\label{obs:contraction}
Let $G$ be a well-partitioned graph with a partition tree $\mathcal{T}$. Let $X,Y$ be two adjacent nodes of $\mathcal{T}$ such that $X\cup Y$ induces a complete subgraph in $G$ and $\mathcal{T}'$ be the tree obtained by contracting the edge $XY$ in $\mathcal{T}$. Now associate the newly created node with the subset of vertices $(X\cup Y)$ and retain all the other nodes of $\mathcal{T}'$ and their associated subsets as in $\mathcal{T}$. Then  $\mathcal{T}'$ is also  a partition tree of $G$. 
\end{observation}
% \begin{proof}
% Let $\mathcal{T}'$ be the tree obtained by contracting the edge $XY$ in $\mathcal{T}$. Now if we associate  and let $R$ be the newly created node. Observe that $R$ has the vertices $(X\cup Y)$. Since, $X\cup Y$ induces a complete subgraph in $G$, the vertices in each node of $\mathcal{T}'$ induce a clique. Now consider any child $C$ of $R$. Observe that $C$ is also a node of $\mathcal{T}$. Let $P(C)$ denote the parent of $C$ in $\mathcal{T}$. Observe that $P(C)$ is either $X$ or $Y$. Since the edges between $C$ and $P(C)$ induce a complete bipartite subgraph, the edges between $C$ and $R$ also induce a complete bipartite subgraph. Therefore, for any edge $AB\in E(\mathcal{T}')$, there exists $A'\subseteq A$ and $B'\subseteq B$ such that the set  $\{uv\in E(G)\colon u\in A', v\in B'\}$ induces a complete bipartite subgraph of $G$. Finally, since no new edges are created due to the contraction of the edge $XY$, for each pair of distinct $A'',B''\in V(\mathcal{T}')$ with $A''B''\notin E(\mathcal{T}')$, there is no edge between a vertex in $A''$ and a vertex in $B''$. Hence, $\mathcal{T}'$ is a partition tree of $G$.
% \end{proof}

 Now we begin building the machinery to describe our algorithm for finding a minimum $(X,Y)$-CVD of $G$ that preserves an $(X,Y)$-edge $ab$. Observe that any $(X,Y)$-CVD set that preserves the edge $ab$ must contain the set $\left(N(a)~\Delta~N(b)\right)$ as subset. (Otherwise, the connected component of $G-S$ containing $ab$ would not be a cluster, a contradiction). 
 
 Let $H$ denote the graph $G-\left(N(a)~\Delta~N(b)\right)$. Now consider the partition $\mathcal{Q}$ defined as $\{Z-\left(N(a)~\Delta~N(b)\right)\colon Z \in V(\mathcal{T})\}$. Now construct a graph $\mathcal{F}$ whose vertex set is $\mathcal{Q}$ and two vertices $Z_1,Z_2$ are adjacent in $\mathcal{F}$ if there is an edge $uv \in E(H)$ such that $u\in Z_1$ and $v\in Z_2$. Observe that $\mathcal{F}$ is a forest.
 
%  We have the following observation.  

% \begin{observation}
% Let $\mathcal{C}$ denote the set of all children and grand children of $X$. Then apart from $X$ and the nodes in $\mathcal{C}$, all other nodes of $\mathcal{T}$ have remain unchanged in $\mathcal{F}$.
% \end{observation}
% \begin{proof}
% Without loss of generality assume
% \end{proof}

Now we have the following observation that relates the connected components of $H$ with that of $\mathcal{F}$ .

\begin{observation}\label{obs:bijection}
There is a bijection $f$ between the connected components of $H$ and the connected components of $\mathcal{F}$, such that for a component $C$ of $H$, $f(C)$ is the partition tree of $C$. Moreover, the vertices of the root node of $f(C)$ is subset of a node in $\mathcal{T}$.
\end{observation}
\begin{proof}
Recall that $\mathcal{Q}$ is a partition of $V(H)$ and the graph $\mathcal{F}$ is a forest. Let $A$ be a connected component of $H$. We have the following cases. 

\begin{enumerate}
    \item\label{it:case-1} There is a vertex $u\in A$ and a vertex in $v\in bd(X,Y)$ such that $uv\in E(G)$. Then observe that $A$ contains both vertices $a$ and $b$. Observe that there is a set $Z=bd(X,Y)$ in $\mathcal{Q}$. Hence, $Z$ is a vertex of $\mathcal{F}$. Now define $f(A)$ to be the subgraph of $\mathcal{F}$ that contains $Z$. Clearly, $f(A)$ is a partition tree of $A$ and the root node of $f(A)$ is $bd(X,Y)$ which is a subset of $X$, the root node of $\mathcal{T}$.
    
    \item There is a vertex $u\in A$ and a vertex in $v\in bd(Y,X)$ such that $uv\in E(G)$. In this case, $A$ contains both vertices $a$ and $b$. Hence, $f(A)$ can be defined as in Case~\ref{it:case-1}.
    
    \item Consider the case when any edge $e=uv$ with $u\in X$ and $v\in A$ satisfies $u\in X-bd(X,Y)$. In this case, observe that $v$ must lie in some child $Z$ of $X$. Moreover, there is a set $Z$ in $\mathcal{Q}$. Hence, $Z$ is a vertex of $\mathcal{F}$. Now define $f(A)$ to be the subgraph of $\mathcal{F}$ that contains $Z$. Clearly, $f(A)$ is a partition tree of $A$ and the root node of $f(A)$ is $Z$ which is a node of $\mathcal{T}$.
    
    \item Consider the case when any edge $e=uv$ with $u\in Y$ and $v\in A$ satisfies $u\in Y-bd(Y,X)$. In this case, observe that $v$ must lie in some child $Z$ of $Y$. Moreover, there is a set $Z$ in $\mathcal{Q}$. Hence, $Z$ is a vertex of $\mathcal{F}$. Now define $f(A)$ to be the subgraph of $\mathcal{F}$ that contains $Z$. Clearly, $f(A)$ is a partition tree of $A$ and the root node of $f(A)$ is $Z$ which is a node of $\mathcal{T}$.
\end{enumerate}
This completes the proof.
\end{proof}

Consider the connected component $H^*$ of $H$ which contains $a$ and $b$ and let $\mathcal{F}'=f(H^*)$ where $f$ is the function given by Observation~\ref{obs:bijection}. Observe that the root $R'$ of $\mathcal{F}'$ is actually $bd(X,Y)$. Moreover, $R'$ has a child $R''$ which is actually $bd(Y,X)$. Observe that, $R'\cup R''$ induces a complete subgraph in $H^*$. Hence, due to Observation~\ref{obs:contraction}, the tree $\mathcal{F}^*$ obtained by contracting the edge $R'R''$ is a partition tree of $H^*$. Moreover, $R^*=R'\cup R'' = bd(X,Y) \cup bd(Y,X)$ is the root node of $\mathcal{F}^*$. Recall that our objective is to find a minimum $(X,Y)$-CVD set that preserves the edge $ab$. We have the following lemma.

\begin{lemma}\label{lem:(X,Y)-CVD-characterisation}
Let $H^*,H_1,H_2,\ldots,H_{k'}$ be the connected components of $H$. Let $S^*$ be a minimum $(R^*)$-CVD set of $H^*$, $S_0=\left(N(a)~\Delta~N(b)\right)$, and for each $j\in [k']$, let $S_j$ denote a minimum CVD set of $H_j$. Then $(S_0 \cup S_1 \cup S_2\cup \ldots \cup S_{k'}  \cup S^*)$ is a minimum $(X,Y)$-CVD set of $G$ that preserves the edge $ab$.
\end{lemma}
\begin{proof}

Observe that, any vertex which is adjacent to $a$ or $b$ lie in $R^*$. Since $S^*$ is a minimum $(R^*)$-CVD set, $ S^* \cap \{a,b\}=\emptyset$ and therefore $H^*-S^*$ has a cluster that contains the edge $ab$. Hence $S_0 \cup S_1\cup S_2 \cup \ldots \cup S_{k''} \cup S^*$ is an $(X,Y)$-CVD set that preserves the edge $ab$. 

Let $Z$ be any $(X,Y)$-CVD set of $G$ that preserves the edge $ab$. For any vertex $u\in S_0$, observe that $a,b,u$ induce a path of length $3$. Hence, $S_0 \subseteq Z$. Let $C$ be a connected component of $G-S_0$. Observe that $Z\cap C$ must be a CVD set of $C$. Therefore, for each $i\in [k']$, $|Z\cap H_i|\leq |S_i|$. 

Since $Z$ is an $(X,Y)$-CVD set of $G$ that preserves the edge $ab$, $\{a,b\}\cap Z=\emptyset$. Since $a,b$ are vertices of $H^*$, $ (Z\cap H^*) \cap \{a,b\} =\emptyset$. Now suppose $(Z\cap H^*)$ is not a $(R^*)$-CVD set of $H^*$. Then due to Lemma~\ref{lem:CVD-characterisation}, $(Z\cap H^*)$ must be a $(R^*,R)$-CVD set of $G^*$ for some child $R$ of $R^*$ in $\mathcal{T}^*$. Hence, there exists a $(R^*,R)$-edge $cd$ which is preserved by $(Z\cap H^*)$. Without loss of generality assume $c\in R^*$ and $d\in R$. Observe that $d$ is not adjacent to $a$ or $b$. Hence, $a,c,d$ induce a path of length $3$ in $H^*-(Z\cap H^*)$, a contradiction. Hence $|Z\cap H^*|\leq |S^*|$. Therefore $|Z| \leq |S_0\cup S_1\cup S_2\cup \ldots \cup S_{k''} \cup S^*|$.
\end{proof}

Lemma~\ref{lem:(X,Y)-CVD-characterisation} provides a way to compute a minimum $(X,Y)$-CVD set of $G$ that preserves the edge $ab$. Clearly, the set $S_0=\left(N(a)~\Delta~N(b)\right)$ can be computed in polynomial time. The following observation provides a way to compute a minimum CVD set of all connected components that are different from $H^*$.

\begin{observation}\label{obs:connected-component-opt}
Let $A$ be a connected component of $H$ which is different from $H^*$. Then a minimum $CVD$ set of $A$ can be computed in polynomial time.  
\end{observation}

\begin{proof}
Recall that $A$ was obtained by deleting $\left(N(a)~\Delta~N(b)\right)$ from $G$, $\mathcal{T}$ is the partition tree of $G$ and root of $\mathcal{T}$ is $X$. Due to Observation~\ref{obs:bijection}, there is a function $f$ between the connected components of $H$ and the connected components of $\mathcal{F}$ such that $f(A)$ is the partition tree of $A$ and there is a node $R\in \mathcal{T}$ such that the vertices in root node of $f(A)$ is a subset of $R$. Now consider the following cases.
\begin{enumerate}
    \item \sloppy Consider the case when $\{a,b\} \cap bd(P(R),R) = \emptyset$.  This implies no vertex of $R$ is adjacent to $a$ or $b$. Moreover, since $A$ is different from $H^*$, $bd(P(R),R) \cap (bd(X,Y)\cup bd(Y,X))=\emptyset$. This further implies that, either $bd(P(R),R)\subseteq N[a]-N[b]$ or $bd(P(R),R)\subseteq N[b]-N[a]$. In either case, $R \cap (N[a]\cup N[b]) = \emptyset$. This implies $R$ is a node of $\mathcal{T}$ distinct from $X$ such that $A$ is isomorphic to $G_R$. Hence, due to the assumption given in Theorem~\ref{thm:(X,Y)-edge}, $OPT(G_R)$ is known and therefore a minimum CVD set of $A$ is known. 
    
    \item Consider the case when there is a vertex $z\in \{a,b\}$ such that $z\in  bd(P(R),R)$. Let $z'$ be the vertex among $a$ and $b$ distinct from $z$. Since $A$ is different from $H^*$,  $z'\not\in bd(P(R),R)$. Hence, $bd(R,P(R)) \subset N(z)$ and therefore $bd(R,P(R)) \subset \left(N(a)~\Delta~N(b)\right)$. This implies that $R$ is a node of $\mathcal{T}$ distinct from $X$ such that $A$ is isomorphic to $G_R - bd(R,P(R))$. Hence, due to the assumption given in Theorem~\ref{thm:(X,Y)-edge}, $OPT(G_R - bd(R,P(R)))$ is known and therefore a minimum CVD set of $A$ is known.
\end{enumerate}
Clearly, distinguishing between the above cases takes $O(|E(G)|)$ time. This completes the proof. 
\end{proof}

Let $H_1,H_2,\ldots, H_{k'}$ be the connected components of $H$, all different from $H^*$. Applying Observation~\ref{obs:connected-component-opt} repeatedly on each component, it is possible to obtain, for each $j\in [k']$, a minimum CVD set $S_j$ of $H_j$. The following observation provides a way to compute a minimum $(R^*)$-CVD set of $H^*$.

\begin{observation}\label{obs:H*}
Let $R$ be a child of $R^*$ in $\mathcal{F}^*$. Then both $OPT(H^*_R)$ and $OPT(H^*_R-bd(R,R^*))$ are known.
\end{observation}

\begin{proof}

Since no vertex of $R$ is adjacent to $a$ or $b$ in $G$, there must exist a node $Q\in \mathcal{T}$ such that the vertices in the node $Q$ is same as that in $R$, $\mathcal{T}_Q=\mathcal{T}^*_R$ and $G_Q=H^*_R$. Moreover, $bd(R,R^*)=bd(Q,P(Q))$, where $P(Q)$ is the parent of $Q$ in $\mathcal{T}$. Hence, due to the assumption given in Theorem~\ref{thm:(X,Y)-edge}, $OPT(H^*_R - bd(R,R^*)$ is known. 
\end{proof}

Due to Observation~\ref{obs:H*} and Theorem~\ref{thm:private}, it is possible to compute a minimum $(R^*)$-CVD set $S^*$ of $H^*$ in $O(|V(G)|\cdot|E(G)|)$ time. Now due to Lemma~\ref{lem:(X,Y)-CVD-characterisation}, we have that $(S_0 \cup S_1 \cup S_2\cup \ldots \cup S_{k'}  \cup S^*)$ is a minimum $(X,Y)$-CVD set of $G$ that preserves the edge $ab$. This completes the proof of Theorem~\ref{thm:(X,Y)-edge} and therefore of Theorem~\ref{thm:(X,Y)}.
In Algorithm~\ref{alg:(X,Y)-CVD-edge}, we give a short pseudocode of our algorithm to find a minimum $(X,Y)$-CVD set of $G$ that preserves an $(X,Y)$-edge $ab$. Using Algorithm~\ref{alg:(X,Y)-CVD-edge}, in Algorithm~\ref{alg:(X,Y)} we provide a short pseudocode to find a minimum $(X,Y)$-CVD set of $G$.

\medskip
 \begin{minipage}{\linewidth}
 \begin{algorithm}[H] 
 \DontPrintSemicolon
\begin{small}
\SetKwInOut{KwIn}{Input}
    \SetKwInOut{KwOut}{Output}
    \KwIn{A well-partitioned chordal graph $G$, a partition tree $\mathcal{T}$ of $G$ rooted at the node $X$, a child node $Y$, an $(X,Y)$-edge $ab$, for each node $Z\in \mathcal{T}-\{X\}$ both $OPT(G_Z)$ and $OPT(G_Z-bd(Z,P(Z)))$ are given as part of input}
    \KwOut{A minimum $(X,Y)$-CVD set of $G$ that preserves the edge $ab$}
    
    Construct the set $S_0=\left(N(a)~\Delta~N(b)\right)$;
    
    Construct the graph $H=G-\left(N(a)~\Delta~N(b)\right)$;
    
    Let $H^*$ be the connected component of $H$ containing $a$ and $b$. Let $H_1,H_2,\ldots,H_{k'}$ be the remaining connected components of $H$.
    
    \For{$i= 1 \text{ to } k'$ }{
    Compute a minimum CVD set $S_i$ of $H_i$ (Observation~\ref{obs:connected-component-opt});
    }
    
    Find the partition tree of $\mathcal{T}^*$ of $G^*$ whose root is $X^*=bd(X,Y) \cup bd(Y,X)$;
    
    Compute a minimum $(X^*)$-CVD set $S^*$ of $G^*$ using Algorithm~\ref{alg:(X)-CVD};
    
    $Sol=S_0 \cup S_1\cup S_2 \cup \ldots \cup S_{k'} \cup S^*$;
    
    \Return Sol;
    % Construct the set ;
    
    % Find a minimum weighted vertex cover $D$ of $\mathcal{H}$;
    
    % Construct the sets $S_1(D), S_2(D), S_3(D)$ and $Sol(D)$ as described in Equations~\ref{eq:sol_1},~\ref{eq:sol_2},~\ref{eq:sol_3} and~\ref{eq:sol}, respectively;
    
%  \Return $Sol(D)$
   \caption{Pseudocode to find a minimum $(X,Y)$-CVD set of a well-partitioned chordal graph that preserves an $(X,Y)$-edge}
   \label{alg:(X,Y)-CVD-edge}
    
\end{small}
 \end{algorithm}
\end{minipage}

\medskip
 \begin{minipage}{\linewidth}
 \begin{algorithm}[H] 
 \DontPrintSemicolon
\begin{small}
\SetKwInOut{KwIn}{Input}
    \SetKwInOut{KwOut}{Output}
    \KwIn{A well-partitioned chordal graph $G$, a partition tree $\mathcal{T}$ of $G$ rooted at the node $X$, a child node $Y$, for each node $Z\in \mathcal{T}-\{X\}$ both $OPT(G_Z)$ and $OPT(G_Z-bd(Z,P(Z)))$ are given as part of input}
    \KwOut{A minimum $(X,Y)$-CVD set of $G$.}
    
    For each $(X,Y)$-edge $e$, compute a minimum $(X,Y)$-CVD set that preserves the edge $e$ using Algorithm~\ref{alg:(X,Y)-CVD-edge};
    
    Let $S$ be a set among all $S_e$'s that has the least cardinality;
    
    \Return $S$;
    % Construct the set ;
    
    % Find a minimum weighted vertex cover $D$ of $\mathcal{H}$;
    
    % Construct the sets $S_1(D), S_2(D), S_3(D)$ and $Sol(D)$ as described in Equations~\ref{eq:sol_1},~\ref{eq:sol_2},~\ref{eq:sol_3} and~\ref{eq:sol}, respectively;
    
%  \Return $Sol(D)$
   \caption{Pseudocode to find $(X,Y)$-CVD set of a well-partitioned chordal graph.}
   \label{alg:(X,Y)}
    
\end{small}
 \end{algorithm}
\end{minipage}

\subsection{Main Algorithm}\label{sec:main-proof}

From now on $G$ denote a fixed well-partitioned chordal graph with a partition tree $\mathcal{T}$ whose vertex set is $\mathcal{P}$, a partition of $V(G)$. We will process $\mathcal{T}$ in the post-order fashion and for each node $X$ of $\mathcal{T}$, we give a dynamic programming algorithm to compute both $OPT(G_X)$ and $OPT(G_X-bd(X,P(X)))$ where $P(X)$ is the parent of $X$ (when exists) in $\mathcal{T}$. Due to Observation~\ref{obs:contraction}, we can assume that $bd(X,P(X)) \subsetneq X$. In the remaining section, $X$ is a fixed node of $\mathcal{T}$, $A$ has a fixed value (which is either $\emptyset$ or $bd(X,P(X))$), $G^A_X$ denotes the graph $G_X - A$. Since well-partitioned chordal graphs are closed under vertex deletion, $G^A_X$ is a well partitioned chordal graph which may be disconnected. Now consider the partition $\mathcal{P}^A$ defined as $\{Y-A\colon Y \in V(\mathcal{T}_X)\}$. Observe that, apart from the set $X$ all other sets of the partitions $\mathcal{P}$ have remained in $\mathcal{P}^A$. Now construct a graph $\mathcal{T}'$ whose vertex set is the partition sets of $\mathcal{P}^A$ and two vertices $X,Y$ are adjacent in $\mathcal{T}'$ if there is an edge $uv \in E(G^A_X)$ such that $u\in X$ and $v\in Y$ (since the graph induced by the union of the sets in $\mathcal{P}^A$ is $G^A_X$, the definition of $\mathcal{T}'$ is valid). Now we have the following observation whose proof is similar to that of Observation~\ref{obs:bijection}.

\begin{observation}\label{obs:bijection-2}
There is a bijection $f$ between the connected components of $G^A_X$ and the connected components of $\mathcal{T}'$, such that for a component $C$ of $G^A_X$, $f(C)$ is a partition tree of $C$, and the root of $f(C)$ is a child of $X$.
\end{observation}

Since the vertices of $X-A$ induces a clique in $G^A_X$, there exists at most one component $G^*$ in $G^A_X$ that contains a vertex from $X-A$.  Due to Observation~\ref{obs:bijection} there exists a unique connected component $f(G^*)=\mathcal{T}^*$ of $\mathcal{T}'$ which is a partition tree of $G^*$. Let the remaining connected components of $G^A_X$ be $G_1,G_2,\ldots, G_k$ and for each $i\in [k]$, let $f(G_i)=\mathcal{T}_i$ and $X_i$ is the root of $\mathcal{T}_i$. Let $X^*$ denote the root node of $\mathcal{T}^*$ and $X_1^*, X_2^*,\ldots,X_t^*$ be the children of $X^*$ in $\mathcal{T}^*$. We have the following observation.

\begin{observation}\label{obs:T*}
% \begin{enumerate}[label=(\alph*)]
    % \item\label{it:1} The root $X^*$ of $\mathcal{T}^*$ is the set $X-A$.
    
    % \item\label{it:2} For each $i\in [k]$, $X_i$ is a child of $X$ in $\mathcal{T}$.
    
    % \item\label{it:3} 
    For each $j\in [t]$, there is a child $Y_j$ of $X$ in $\mathcal{T}$ such that $Y_j=X^*_j$ and $G_{Y_j} = G^*_{X^*_j}$.
% \end{enumerate}
\end{observation}

\begin{proof}
Observe that the root of $\mathcal{T}^*$ is $X^*=X-A$. Since $A\subsetneq X$, any child of $X^*$ must be a child of $X$.
\end{proof}

We have the following lemma.

\begin{lemma}\label{lem:optimum}
$OPT(G^A_X) =  \left( \displaystyle\bigsqcup\limits_{i=1}^{k} OPT(G_{X_i}) \right) \sqcup OPT(G^*)$
\end{lemma}

\begin{proof}
The lemma follows directly from the fact that $G_{X_1},G_{X_2},\ldots,G_{X_k}$ and $G^*$ are connected components of $G^A_X$. 
\end{proof}

% Recall that $X^*$ is the root of $\mathcal{T}^*$ and let $X^*_1,X^_2,\ldots,Y_{k''}$ be the 
Due to Observation~\ref{obs:bijection-2}, $OPT(G_{X_i})$ is already known. Due to Lemma~\ref{lem:(X,Y)-CVD-characterisation}, any CVD set $S$ of $G^*$ is either a $(X^*)$-CVD set or there exists a unique child $Y$ of $X^*$, such that $S$ is a $(X^*,Y)$-CVD set of $G^*$. by Theorem~\ref{thm:private}, it is possible to compute a minimum $(R^*)$-CVD set $S_0$ of $G^*$. Due to Observation~\ref{obs:T*}, for any node $Y$ of $\mathcal{T}^*$ which is different from $X^*$, both $OPT(G_{Y})$ and $OPT(G_{Y} - bd(Y,P(Y)))$ are known, where $P(Y)$ is the parent of $Y$ in $\mathcal{T}^*$. Hence, by Theorem~\ref{thm:(X,Y)} for each child $X^*_i$, $i\in [t]$, computing a minimum $(X^*,X^*_i)$-CVD set $S_i$ is possible in $O(|V(G^*_{X^*_i})|\cdot |E(G^*_{X^*_i})|)$ time. Let $S^*\in \{S_0,S_1,S_2,\ldots,S_t\}$ be a set with the minimum cardinality. Due to Lemma~\ref{lem:CVD-characterisation}, $S^*$ is a minimum CVD set of $G^*$ that can be obtained in $O(m^{2}n)$. Finally, due to Lemma~\ref{lem:optimum}, we have a minimum CVD set of $G^A_X$.% \input{interval3}

\section{$O(n(n+m))$-time algorithm for $s$-CVD on interval graphs}\label{sec:thm-interval}

In this section we shall give an $O(n(n+m))$-time algorithm to solve \textsc{$s$-CVD} on interval graph $G$ with $n$ vertices and $m$ edges. For a set $X \subseteq V(G)$, if each connected component of $G-X$ is an $s$-club, then we call $X$ as an \emph{$s$-club vertex deleting set} ($s$-CVD set). In the next section we present the main ideas of our algorithm to find a minimum cardinality \textsc{$s$-CVD} set of an interval graph.

\subsection{Overview of the algorithm}

In the heart of our algorithm lies a characterisation of \textsc{$s$-CVD} sets of an interval graph. We show (in Lemma~\ref{lem:type}) that any $s$-CVD set must be one of four types, defined in Definitions~\ref{def:type-1}- \ref{def:type-4}. Hence, the problem boils down to computing a minimum $s$-CVD set of each type. To do this, first we arrange the maximal cliques in the order of its \emph{Helly region}. Let $Q_1,Q_2,\ldots,Q_k$ be the ordering of the cliques. Then for each $1\leq a\leq k$, we find minimum cardinality $s$-CVD set of the graph $\Inducedgraph{G}{1}{a}$ which is the subgraph induced by the vertices in $(Q_1\cup Q_2\cup \ldots\cup Q_a)$. Moreover, to facilitate future computations we also find minimum $s$-CVD set of the graph $\Inducedgraph{G}{1}{a} - A$ where $A=Q_a\cap Q_b$ for some $a<b\leq k$. The trick was to show that, by solving \textsc{$s$-CVD} on $O(n+m)$ many different ``induced subgraphs" of $G$, it is possible to solve \textsc{$s$-CVD} on $G$. In other words, by solving $O(n+m)$ many different subproblems, it is possible  to solve \textsc{$s$-CVD} on $G$. Moreover, it is possible to solve a subproblem in $O(n)$ time. In Section~\ref{sec:interval-notation} we define four types of $s$-CVD sets and state that any optimal solution must be one of those four types. In Section~\ref{sec:algorithm} we give a sketch of our algorithm and analyse the time complexity in Section~\ref{sec:interval-time}.

\subsection{Definitions and main lemma}\label{sec:interval-notation}

Let $G$ denotes a connected interval graph with $n$ vertices and $m$ edges. The set $\mathcal{I}$ denotes a fixed interval representation of $G$ where the endpoints of the representing intervals are distinct. Let $\leftend{v}$ and $\rightend{v}$ denote the left and right endpoints, respectively, of an interval corresponding to a vertex $v \in V(G)$. Then the interval assigned to the vertex $v$ in $\mathcal{I}$ is denoted by $\Interval{v}=\left[\leftend{v},\rightend{v}\right]$.

 Observe that, intervals on a real line satisfies the  Helly property and hence  for each maximal clique $Q$ of $G$ there is an interval $I = \displaystyle\bigcap_{v\in Q} \Interval{v}$. We call $I$ as the \helly region corresponding to the maximal clique $Q$. Let $Q_1,Q_2,\ldots,Q_k$ denote the set of maximal cliques of $G$ ordered with respect to their \helly regions $I_a, 1 \leq a \leq k$ on the real line. That is, $I_1 < I_2 <\ldots < I_k$. Observe that, for any two integers $a,b$ we have $I_a \cap I_b = \emptyset$ as both $Q_a$ and $Q_b$ are maximal cliques. Moreover, for any $a\leq b \leq c$ if a vertex $v\in Q_a\cap Q_c$, then $v\in Q_b$.\\
 
With respect to an ordering of maximal cliques $Q_1,Q_2,\ldots,Q_k$ of $G$, we define the following. \\
\begin{definition} \label{def:clique_ordering}
\begin{enumerate} [label=(\roman*)]
    \item For integers a,b where $1\leq a < b \leq k$, let $\Sep{a}{b} = Q_a\cap Q_b$.
    \item For an integer $a$, let $$\Allseparate{Q_a}=\left\{ \Sep{a}{b} \colon a < b \leq k \text{ and } \Sep{a}{b} \neq \Sep{a}{b'}, a < b'< b \right \} \cup \emptyset$$.
    {\footnotesize(Note that, the members of the set $\Allseparate{Q_a}$ are distinct.) }
 \item  For $ A \in \Allseparate{Q_a}$, let $Y_A^a = (Q_a-Q_{a-1})-A$.
 \item For a vertex $v\in V(G)$, the index $\minIndex{v}=\min\{a\colon v\in Q_a\}$. That is, the minimum integer $a$ such that $v$ belongs to the maximal clique $Q_a$.
 \item  For a vertex $v\in V(G)$, the index $ \maxIndex{v} = \max \{a \colon v\in Q_a\}$. That is, the maximum integer $b$ such that $v$ belongs to the maximal clique $Q_b$.
\end{enumerate}

\end{definition}

  We use the following observation to prove our main lemma.

 \begin{observation}\label{obs:consecutive_seprataion}

 Let $X\subseteq V(G)$ and $u,v$ be two vertices with $\rightend{u} < \leftend{v}$ such that  $u$ and $v$ lie in different connected components in $G-X$.  Then there exists an integer $a$ with $ \maxIndex{u} \leq a < \minIndex{v}$, such that  $\Sep{a}{a+1} \subseteq X$.
\end{observation}
 
\begin{proof}

Let $\mathcal{C}$ be the set of all connected components of $G-X$. For a connected component $C\in \mathcal{C}$, define $\hat{r}(C)=\max\{\rightend{v}\colon v\in C\}$ and $\hat{l}(C)=\min\{\leftend{v} \colon v\in C\}$. Note that the interval $[\hat{l}(C),\hat{r}(C)] = \bigcup\limits_{v \in V(C)}I(v)$ and we call it as the \emph{span}($C$). Observe that for two distinct connected components $C, C' \in \mathcal{C}$ we have $\emph{span}(C)\cap \emph{span}(C') = \emptyset$. Therefore, $\mathcal{C}$ can be ordered with respect to the order in which the span of components appears on the real line. Let $C_1,\ldots,C_x$ be this ordering. We define \emph{gap}$(C_i,C_{i+1}) =(\hat{r}(C_i),\hat{l}(C_{i+1})), 1 \leq i \leq x-1$. Note that any vertex whose corresponding interval contains a point in \emph{gap}$(C_i,C_{i+1})$ should be a member of $X$: otherwise that vertex belongs to another component in between $C_i$ and $C_{i+1}$ (by definition of \emph{gap}$(C_i,C_{i+1})$) which contradicts the ordering of components. Let $C^u = C_t$ and $C^v=C_{t'}$ denote the connected components of $G-X$ that contain $u$ and $v$, respectively. Since $\rightend{u} < \leftend{v}$, we have $t < t'$.
%Observe that there exists $C\in \mathcal{C}$ with $\hat{r}(C_u) \leq  \hat{r}(C) < \hat{l}(C_v)$ such that there is no connected component $C'\in \mathcal{C}$ with $\hat{r}(C) < \hat{r}(C') < \hat{l}(C_v)$.Let $C^*\in \mathcal{C}$ be the connected component with the maximum value of $\hat{r}(C')$ : $C' \in \mathcal{C} \text{ and } \hat{r}(C') \in [\hat{r}(C_u), \hat{l}(C_v))$.

Let $p\in V(G)$ be such that $\rightend{p} = \max\{\rightend{w}\colon w\in V(G), \rightend{w}<\hat{l}(C^v)\}$. %$\rightend{q} = \min\{\rightend{w} \colon w\in V(G), \rightend{w} > \hat{l}(C^v)\}$
Now take $a = \maxIndex{p}$, the maximum index $i$ such that $p \in Q_i, 1 \leq i \leq k$. For the index $a$, we will show that $\maxIndex{u} \leq a < \minIndex{v}$ and  $\Sep{a}{a+1} \subseteq X$.\\

\textbf{(i) $\maxIndex{u} \leq a < \minIndex{v}$}: It is immediate from the definition of $\rightend{p}$ that $\rightend{u} \leq \hat{r}(C^{u}) \leq \rightend{p}$ and $\rightend{p} <\hat{l}(C^v) \leq \leftend{v}$. Since $\rightend{p} <\leftend{v}$, observe that  the \helly region corresponding to the clique containing the vertex $p$ come before that of $v$ on the real line. Moreover, 
since the maximal cliques are numbered with respect to  the order in which their  \helly regions appear on the real line, we can infer that $ \maxIndex{p}=a < \minIndex{v}$. Similarly, since $\rightend{u} \leq \rightend{p}$, by similar arguments as above, we have $\maxIndex{u} \leq a$. Therefore we have proved $\maxIndex{u} \leq a < \minIndex{v}$ \\

\textbf{(ii)$\Sep{a}{a+1} \subseteq X$}: Consider the component $C_{t'-1}$ which comes in the immediate left of $C^v$ in the ordering of the components in $\mathcal{C}$. Since $\rightend{p} <\hat{l}(C^v)$, the \helly region of $Q_a$ ends before $span(C^v)$. Observe that $\rightend{p} \geq \hat{r}(C_{t'-1})$. Moreover, the \helly region of $Q_{a+1}$ starts after that of  $Q_{a}$. Since $p \notin Q_{a+1}$ by definition of $a$ it follows that \helly region of $Q_{a+1}$ is after the  $span(C_{t'-1})$. Therefore, the intervals corresponding to those vertices common to both $Q_a$ and $Q_{a+1}$ contain some points of \emph{gap}$(C_{t'-1},C^v)$. This implies $\Sep{a}{a+1} \subseteq X$. 
\end{proof}

For two integers $a,b$ with $1\leq a\leq b\leq k$, let $\Inducedgraph{G}{a}{b}$ denotes the subgraph induced by the set $\{Q_a \cup Q_{a+1} \cup \ldots \cup Q_b\}$. 

\begin{definition}
For an induced subgraph $H$ of $G$, a vertex $v\in V(H)$ and an integer $a$, let $\Level{H}{a}{v}$ denote the set of vertices in $H$ that lie at distance $a$ from $v$ in $H$. 
\end{definition}

In the remainder of this section, we  use the notation $\Level{H}{s+1}{v}$ where $H=\Inducedgraph{G}{1}{a}- A$ for some integer $a$ and $v\in Y^a_A$  (See Definition \ref{def:clique_ordering}, (iii)) several times.

\begin{definition}
For an integer $a, 1\leq a\leq k-1$ and a set $A \in \Allseparate{Q_a}$ consider the induced subgraph $H=\Inducedgraph{G}{1}{a}- A$ and the sub-interval representation $\mathcal{I}'\subseteq \mathcal{I}$ of $H$. We define the ``frontal component" of the induced graph  as the connected component of $\Inducedgraph{G}{1}{a}- A$ containing the vertex with the rightmost endpoint in $\mathcal{I}'$. 
\end{definition}

Note that for an integer $a$ and $A\in \Allseparate{Q_a}$, the vertices of $Y_A^a$, if any,  lies in the \emph{frontal} component of $\Inducedgraph{G}{1}{a}- A$. Below we categorize an $s$-CVD set $X$ of $\Inducedgraph{G}{1}{a}-A$ into four types. In the following definitions, we consider an integer $a, 1<a\leq k$ and a set $A\in \Allseparate{Q_a}$.
\begin{definition}\label{def:type-1}
 An $s$-CVD set $X$ of $\Inducedgraph{G}{1}{a}-A$ is of ``type-$1$'' if  $Y_A^a \subseteq X$.  \end{definition}

% If  an $s$-CVD set $X$ of $\Inducedgraph{G}{1}{a}-A$ is not of ``type-$1$'' then, by definition,  there exists a vertex $v$ in $\Inducedgraph{G}{1}{a}-A$ such that $v\in Y_A^a$. We consider such a vertex $v$ and have the following definitions. 

%   Let $H=\Inducedgraph{G}{1}{a}-A$. Note that when $Y_A$

\begin{definition}\label{def:type-2}
 An $s$-CVD set $X$ of $H=\Inducedgraph{G}{1}{a}-A$ is of ``type-$2$'' if there is a vertex $v\in Y_A^a$ such that $\Level{H}{s+1}{v} \subseteq X$.  
\end{definition}

\begin{definition}\label{def:type-3}
 An $s$-CVD set $X$ of $H=\Inducedgraph{G}{1}{a}-A$ is of ``type-$3$'' if there exists an integer $c, 1\leq c < a$ such that $\Sep{c}{c+1}-A \subseteq X$ and $\Inducedgraph{G}{c+1}{a}-(\Sep{c}{c+1}\cup A)$ is connected and has diameter at most $s$.
 \end{definition}

% Consider an integer $a, 1<a\leq k$, a set $A\in \Allseparate{Q_a}$ and a vertex $v\in Y_A^a$.
\begin{definition}\label{def:type-4}
 An $s$-CVD set $X$ of $H=\Inducedgraph{G}{1}{a}-A$ is of ``type-$4$'' if there exists an integer $c, 1\leq c < a$ such that $\Sep{c}{c+1}-A \subseteq X$ and $\Inducedgraph{G}{c+1}{a}-(\Sep{c}{c+1}\cup A)$ is connected and has diameter exactly $s+1$.
 \end{definition}

The following lemma is crucial for our algorithm.

\begin{lemma}[Main Lemma]\label{lem:type}
Consider an integer $1\leq a\leq k$ and a set $A\in \Allseparate{Q_a}$. Then at least one of the following holds:
\begin{enumerate}
    \item Every connected component of $\Inducedgraph{G}{1}{a}-A$ have diameter at most $s$.
    
    \item Any $s$-CVD set of $\Inducedgraph{G}{1}{a}-A$ is of some type-$j$ where $j\in \{1,2,3,4\}$.
\end{enumerate}
\end{lemma}

\begin{proof}
Assume that the \emph{frontal} component of $H=\Inducedgraph{G}{1}{a}-A$ has diameter at least $s+1$ and the set $Y_A^a \neq \emptyset$. Otherwise, any $s$-CVD set $X$ of $H$ is of either type-$1$  or type-$2$: Type-$1$ is obvious when $Y_A^a = \emptyset$ because $\emptyset \subseteq X$. If the diameter  of \emph{frontal} component is at most $s$ then the set  $\Level{H}{s+1}{v} = \emptyset$ and hence any $s$-CVD set of $H$ is of type-$2$.\\
Let $H$ has an $s$-CVD set $X$ that is not of type-$j$ for any $j\in \{1,2\}$ and $v$ be a vertex in $Y_A^a$. Since $X$ is not of type-2, $H-X$ contains a vertex $u$ such that $u\in \Level{H}{s+1}{v}$. Now choose a vertex $u\in \Level{H}{s+1}{v}$ such that $\maxIndex{u} = \max\{\maxIndex{u'} \colon u'\in \Level{H}{s+1}{v}-X \} $.\\

Let $X'=X \cup A$. Then observe that $G-X' = H-X$ and hence, $ u \in G-X'$. Since $X$ is an $s$-CVD set of $H$  and the distance between $u$ and $v$ in $H$ is $s+1$, the vertices $u$  and $v$ must lie in different connected components in $G-X'$. Therefore, by Observation~\ref{obs:consecutive_seprataion}, there is an integer $b$ such that $\Sep{b}{b+1} \subseteq X'$ and $\maxIndex{u} \leq b < \minIndex{v}$. Let $b$ be the maximum among all $b'$ such that 
$\maxIndex{u} \leq b' < \minIndex{v}$ and $\Sep{b'}{b'+1} \subseteq X'$. Note that $\Sep{b}{b+1} \subseteq X'$ implies $\Sep{b}{b+1} -A \subseteq X$. To complete the proof we need the following claim.
\begin{myclaim}\label{clm:connected}
Let $Y$ be a subset of $H$ such that $\Sep{b}{b+1} \subseteq Y \subseteq X$ where $b$ is the maximum among all $b'$ such that 
$\Sep{b'}{b'+1} \subseteq X$. Then $\Inducedgraph{G}{b+1}{a}-(Y\cup A)$ is connected.
\end{myclaim}
\begin{proof}[Proof of Claim:]
Suppose $\Inducedgraph{G}{b+1}{a}-(Y\cup A)$ is not connected. Let $Z=Y\cup A$ and $C_v$ be the connected component containing a vertex $v \in Q_a$ (Note that $Y_A^a \neq \emptyset)$) in $\Inducedgraph{G}{b+1}{a}-Z$.  Since $\Inducedgraph{G}{b+1}{a}-Z$ is not connected, there exists a vertex $u' \in G-Z$ such that $u' \not \in C_v$. Let $C_{u'}$ be the connected component containing $u'$. Observe that $\maxIndex{u'} < \minIndex{v} = a$ and $G-Z$ is also not connected. Hence by Observation \ref{obs:consecutive_seprataion}, there exists an integer $b^*$ such that $\Sep{b^*}{b^*+1} \subseteq Z$ and $\maxIndex{u'} \leq b^* < \minIndex{v}$. Since $u' \in \Inducedgraph{G}{b+1}{a}- Z$, the index $\maxIndex{u'} > b$. Thus it follows that $b  < \maxIndex{u} \leq b^* < \minIndex{v}$, which contradicts the maximality of the index $b$.
\end{proof}
Let $H_b = \Inducedgraph{G}{b+1}{a}-(\Sep{b}{b+1}\cup A)$.  Now we show that $H_b$ has diameter at most $s+1$. Otherwise, $H_b$ contains vertices that are at distance greater than $s+1$ from the vertex $v$. Let $Q_{b''}$ be the highest indexed maximal clique containing a vertex $x$ such that distance between $x$ and $v$ in $H_b$ is exactly $s+2$. Observe that $b'' >b$. Now we show that $\Sep{b''}{b''+1} \subseteq X$ which contradicts the maximality of $b$ (See the definition of $b$ defined in the above paragraph.) \\

For that, since $\Sep{b''}{b''+1} \subseteq Q_{b''+1}$,  the maximality of $b''$ implies that the vertices in  $\Sep{b''}{b''+1}$ are at distance $s+1$ from $v$ in $H_b$. Note that by the above claim, the induced subgraphs  $H_b$ and  $\Inducedgraph{G}{b+1}{a}-(X \cup A)$ are connected. Moreover, since $X$ is an $s$-CVD set of $H=\Inducedgraph{G}{1}{a}-A$, when $\Sep{b}{b+1} -A \subseteq X$ all vertices at distance greater than $s+1$ from the vertex $v$ in $H_b$ must be in $X - (\Sep{b}{b+1}-A)$. Therefore, $\Sep{b''}{b''+1} \subseteq X$ and $\Sep{b''}{b''+1} \cup A \subseteq X'$. This contradicts the maximality of $b$.
If the diameter of $H_b$  is exactly $s+1$, then $X$ is of type-$4$. Otherwise, $X$ is of type-$3$. 
\end{proof}

\subsection{Some more observations}

Let $H$ be an induced subgraph of $G$ and $u,v$ be two vertices of $H$. The distance between $u$ and $v$ in $H$ is denoted by $\dist{H}{u}{v}$. 

\begin{observation}\label{obs:paired-distance}
Consider two integers $a,b$ with $1\leq a<b\leq k$ and a set $A \in \Allseparate{Q_b}$. Let $H = \Inducedgraph{G}{1}{b}-A$ and $u,v,w$ be three vertices of $H$ such that $\{u,v\} \subseteq Q_{b} - Q_{b-1}$ and $w \in Q_a$. Then $\dist{H}{u}{w} = \dist{H}{v}{w}$.
\end{observation}

\begin{proof}
Suppose for contradiction that $\dist{H}{u}{w} \neq \dist{H}{v}{w}$. Without loss of generality assume that $\dist{H}{u}{w} < \dist{H}{v}{w}$. Let $P$ be a shortest path between $u$ and $w$ in $H$ and $u'$ be the vertex in $P$ which is adjacent to $u$. Observe that $u'\in Q_b\cap Q_{b-1}$(this is because: $u$  is not intersecting with the \helly region of $Q_{b-1}$,  $a < b$ in the ordering and P is a shortest path). Therefore $u'$ is adjacent to $v$ and  $P'=(P-\{u\})\cup \{v\}$ is a path between $v$ and $w$ such that $\dist{H}{v}{w} \leq |P'|=|P| = \dist{H}{u}{w}$, a contradiction. 
\end{proof}

\begin{observation}\label{obs:frontal component}
Let $C_f^H$  be the \emph{frontal} component of $H=\Inducedgraph{G}{1}{a}-A^{*}, A^{*} \subseteq V(G)$. Let $Y_{A^*}^a= (Q_a-Q_{a-1})-A^*$. If $Y_{A^*}^a \neq \emptyset$ then any vertex $v \in Y_{A^*}^a$ is an end vertex of a diametral path (a shortest path whose length is equal to the diameter of a graph) of $C_f^H$. 
\end{observation}
\begin{proof}
 Suppose that $Y_{A^*}^a \neq \emptyset$ and no vertex $v \in Y_{A^*}^a$ is an end vertex of a diametral path of $C_f^H$. Let $P$ be a diametral path of $C_f^H$ and $x,y$ be the end vertices. Observe that neither $x$ nor $y$ is in $Y_{A^*}^a$. Without loss of generality assume that $\minIndex{x} \leq \minIndex{y}$. Let $P'$ be a shortest path between $x$ and $v$ where $v \in Y_{A^*}^a$. Since $P$ has the maximum size among the shortest paths and $P^{'}$ is not a diametral path, we have $|P'|< |P|$.  Since $v \in Y_{A^*}^a$  and $x, y \notin Y_{A^*}^a$ we have $a=\minIndex{v} > \minIndex{y} \geq \minIndex{x}$. Hence the path $P'$ contains a vertex $w$ such that $w \neq v$ and $\minIndex{w} \leq \minIndex{y}\leq \maxIndex{w}$ (That is, any path from $v$ to $x$ should cross the cliques containing $y$). This implies $w$ is a neighbor of $y$ and there exists a path $P''$ between $x$ and $y$ via $w$ such that $|P''| \leq |P'|$ (the path $P''$ is obtained by adding the edge $wy$ to the subpath  from $x$ to $w$ in $P'$). Since $|P'|< |P|$, this contradicts the assumption that $P$ is a shortest path between $x$ and $y$. Therefore, there exists at least one vertex $v \in Y_{A^*}^a$ which is an end vertex of a diametral path of $C_f^H$. Then by Observation \ref{obs:paired-distance}, each vertex in  $Y_{A^*}^a$ is an end vertex of a diametral path of $C_f^H$.
\end{proof}

\subsection{The algorithm}\label{sec:algorithm}
Our algorithm constructs a table $\OPTT$ iteratively whose cells are indexed by two parameters. For an integer $a, 1\leq a\leq k$ and a set $A \in \Allseparate{Q_a}$, the cell $\OPT{a}{A}$ contains a minimum \textsc{$s$-CVD} set of $\Inducedgraph{G}{1}{a}- A$. Clearly, $\OPT{k}{\emptyset}$ is a minimum $s$-CVD set of $G$. 

Now we start the construction of $\OPTT$. Since $\Inducedgraph{G}{1}{1}$ is a clique, we set $\OPT{1}{A} = \emptyset$ for all $A\in \Allseparate{Q_1}$:  

\begin{lemma}\label{lem:recursion-base}
For any $A\in \Allseparate{Q_1}$, $\OPT{1}{A} = \emptyset$.
\end{lemma}
From now on assume $a\geq 2$ and $A$ be a set in $\Allseparate{Q_a}$. \
Let $H$ be the graph $\Inducedgraph{G}{1}{a}-A$ and $F$ be the graph $\Inducedgraph{G}{1}{a-1}-(A\cap Q_{a-1})$. Observe that for any two integers $a, b, 1 \leq a < b \leq k$ the set $S_{a-1}^{b} = S_{a}^{b} \cap Q_{a-1}$. Then,  for any $A\in \Allseparate{Q_a}$ we have $(A\cap Q_{a-1})\in \Allseparate{Q_{a-1}}$ and $\OPT{a-1}{A\cap Q_{a-1}}$ is defined. Note that $H-F = Y_A^a$.

In the following lemma we show that $\OPT{a}{A} = \OPT{a-1}{A\cap Q_{a-1}}$ if the \emph{frontal} component of $H$ has diameter at most $s$.
% If $Y_A^a = \emptyset$ then the \emph{frontal} components of $H$ is the frontal component of $F$ itself and hence, $\OPT{a}{A} = \OPT{a-1}{A\cap Q_{a-1}}$. When $Y_A^a \neq \emptyset$, the connected components of $H$ and $F$ are same except the \emph{frontal} components.
\begin{lemma}\label{lem:recursion-1}
Let $H=\Inducedgraph{G}{1}{a}-A$, for $A\in \Allseparate{Q_a}, 1 < a \leq k$. If the \emph{frontal} component of $H$ has diameter at most $s$, then $\OPT{a}{A} = \OPT{a-1}{A\cap Q_{a-1}}$.
\end{lemma}
\begin{proof}
Let $F$ denote the graph $\Inducedgraph{G}{1}{a-1}-(A\cap Q_{a-1})$. Since $H= F \cup Y_A^a$, if $Y_A^a = \emptyset$ then $H=F$ and hence, $\OPT{a}{A} = \OPT{a-1}{A\cap Q_{a-1}}$. Now assume that $Y_A^a \neq \emptyset$. Observe that the connected components of $H$ and $F$ are same except the \emph{frontal} components. The \emph{frontal} components of $H$ and $F$ differs depending on the set $S_{a-1}^{a}$ as follows. 
\begin{enumerate}
      \item [i)] If $S_{a-1}^{a} \cap H = \emptyset$ then the \emph{frontal} component of $H$ is $Y_A^a$.
    \item[ii)] If $S_{a-1}^{a} \cap H \neq \emptyset$ then the \emph{frontal} component of $H$ is the union of the \emph{frontal} component of $\Inducedgraph{G}{1}{a-1}-A$ and $Y_A^a$.
   \end{enumerate}
   If the \emph{frontal} component of $H$ is $Y_A^a$ then $\OPT{a}{A} = \OPT{a-1}{A\cap Q_{a-1}}$ because diameter of $Y_A^a$ is $1$. Hence assume that the \emph{frontal} component of $H$ belongs to the case (ii) defined above. Let $C_f^H$ be the \emph{frontal} component of $H$ and $C_f^F$ be the \emph{frontal} component of $F$. Then $C_f^H= C_f^F \cup Y_A^a$. We have the following claim.
 
   \begin{myclaim}\label{clm:frontal component}
Let $C_f^H= C_f^F \cup Y_A^a$. If the diameter of $C_f^H$ is at most $s$ then the diameter of $C_f^F$ is also at most $s$.
\end{myclaim}
\begin{proof}[Proof of Claim:]
Suppose not, then $C_f^F$ contains two vertices $u$ and $v$ such that the distance between $u$ and $v$ in $C_f^F$ is at least $s+1$. Without loss of generality, assume that $\leftend{u} < \leftend{v}$. Let $P$ be a shortest path between $u$ and $v$ in $C_f^F$. Observe that since $C_f^F=C_f^H - Y_A^a$, no vertex $w\in Y_A^a$ belongs to $V(P)$. Moreover, for any vertex  $w\in Y_A^a$ we have $\leftend{u} < \leftend{v} < \leftend{w}$ in the interval representation. Therefore, any shortest path between $u$ and $v$ in $C_f^H$ does not contain a vertex $w\in Y_A^a$. Hence the shortest path between $u$ and $v$ in $C_f^H$ is also at least $s+1$ which contradicts the assumption that the diameter of $C_f^H$ is at most $s$.
\end{proof}
 Hence by the minimality of $\OPT{a-1}{A\cap Q_{a-1}}$, no vertices of $C_f^F$ are in $\OPT{a-1}{A\cap Q_{a-1}}$. Thus it follows that $\OPT{a}{A} = \OPT{a-1}{A\cap Q_{a-1}}$.
\end{proof}

Now assume that the \emph{frontal} component of $H=\Inducedgraph{G}{1}{a}-A$ has diameter at least $s+1$.  Recall that if  $Y_A^a = \emptyset$, we have $\OPT{a}{A} = \OPT{a-1}{A\cap Q_{a-1}}$. Hence assume that  $Y_A^a \neq \emptyset$. Due to Lemma~\ref{lem:type}, any $s$-CVD set of $H$ has to be one of the four types defined in Section~\ref{sec:interval-notation}.

 First, for each $j\in \{1,2,3,4\}$, we find an $s$-CVD set of minimum cardinality, which is of type-$j$. We begin by showing how to construct a minimum cardinality $s$-CVD set $X_1$  of type-$1$ of $\Inducedgraph{G}{1}{a}-A$. We define $X_1$ as below.
 
%  define four sets $X_1,X_2,X_3,X_4$ as follows.

\begin{equation}\label{eq1}
X_1 = Y^a_A\cup \OPT{a-1}{A\cap Q_{a-1}}
\end{equation}

\begin{lemma}
The set $X_1$ is a minimum cardinality $s$-CVD set of type-$1$ of $\Inducedgraph{G}{1}{a}-A$.
\end{lemma}
\begin{proof}
Observe that the graph $H= \Inducedgraph{G}{1}{a}-Y^a_A$ is isomporphic to $\Inducedgraph{G}{1}{a-1}-(A \cap Q_{a-1})$. Hence $ X_1 = Y^a_A \ \cup  \OPT{a-1}{A\cap Q_{a-1}}$ is an $s$-CVD set of $H$. By definition, $Y^a_A$ is included in an $s$-CVD set of type-$1$. Hence the minimality of $\OPT{a-1}{A\cap Q_{a-1}}$ implies that $X_1$ is a minimum cardinality set of type-$1$.
\end{proof}

%       \item[$\square$] \textbf{Type-$1$ $s$-CVD set}\\ \begin{equation}\label{eq1}
%     X_1 = Y^a_A\cup \OPT{a-1}{A\cap Q_{a-1}}
% \end{equation}
% Observe that the graph $H= \Inducedgraph{G}{1}{a}-Y^a_A$ is isomporphic to $\Inducedgraph{G}{1}{a-1}-(A \cap Q_{a-1})$. Hence $ X_1 = Y^a_A \ \cup  \OPT{a-1}{A\cap Q_{a-1}}$ is an $s$-CVD set of $H$. By definition, $Y^a_A$ is included in an $s$-CVD set of type-$1$. Hence the minimality of $\OPT{a-1}{A\cap Q_{a-1}}$ implies that $X_1$ is a minimum cardinality set of type-$1$.\\
  Let $v$ be some vertex in $Y^a_A$ and $b<a$ be the maximum integer such that \sloppy  $\left(Q_{b} \cap \Level{H}{s+2}{v}\right) \neq \emptyset$. We construct a minimum cardinality $s$-CVD set of type-$2$ of $\Inducedgraph{G}{1}{a}-A$ defined as follows.

\begin{equation}\label{eq2}
    % X_2 = \Level{H}{s+1}{v} \cup \OPT{b}{Q_b\cap Q_{b+1}}
    X_2 = \Level{H}{s+1}{v} \cup \OPT{b}{\Sep{b}{b+1}}
\end{equation}
\begin{lemma}
The set $X_2$ is a minimum cardinality $s$-CVD set of type-$2$ of $\Inducedgraph{G}{1}{a}-A$.
\end{lemma}
\begin{proof}
 By the maximality of $b$ we have $\Sep{b}{b+1} \subseteq \Level{H}{s+1}{v}$. Moreover, the graph $(\Inducedgraph{G}{b+1}{a}-A)-\Level{H}{s+1}{v}$ is connected: otherwise, if $\Level{H}{s+1}{v}$ is a separator of $(\Inducedgraph{G}{b+1}{a}-A)$ then $(\Sep{b'}{b'+1}-A) \subseteq \Level{H}{s+1}{v}$ for some $b' > b$. Since $Q_{b'}$ is a maximal clique there exists at least one vertex $w \in Q_{b'}$ and $w \not \in Q_{b'+1}$. Hence the distance between $w$ and $v$ is $s+2$ and $\left(Q_{b'} \cap \Level{H}{s+2}{v}\right) \neq \emptyset$. Since $b' >b$, this contradicts the maximality of $b$.\\
  Since $(\Inducedgraph{G}{b+1}{a}-A)-\Level{H}{s+1}{v}$ is connected we have $((\Inducedgraph{G}{b+1}{a}-A)-\Level{H}{s+1}{v})$ is a frontal component of $G[1,a]- (A \cup \Level{H}{s+1}{v}$. Let $A'=A \cup \Level{H}{s+1}{v}$. Note that $Y_{A'}^a = Y_{A}^a \neq \emptyset$.  Observe that the distance between $v \in Y_{A'}^a$ and any other vertex in $(\Inducedgraph{G}{b+1}{a}-A)-\Level{H}{s+1}{v}$ is at most $s$. Hence by Observation \ref{obs:frontal component}, $(\Inducedgraph{G}{b+1}{a}-A)-\Level{H}{s+1}{v}$ has diameter at most $s$.\\
  Note that any vertex of $\Inducedgraph{G}{1}{b}$ that belongs to $A$ is also in $\Sep{b}{b+1}$. Hence $\Inducedgraph{G}{1}{b} - (A \cup \Sep{b}{b+1}) = \Inducedgraph{G}{1}{b} - \Sep{b}{b+1}$. Since $\OPT{b}{\Sep{b}{b+1}}$ is a minimum cardinality $s$-CVD set of $\Inducedgraph{G}{1}{b} - \Sep{b}{b+1}$ the set $X_2 = \Level{H}{s+1}{v} \cup \OPT{b}{\Sep{b}{b+1}}$ is an $s$-CVD set of $H$. By definition, $\Level{H}{s+1}{v} $ is included in an $s$-CVD set of type-$2$. Observe that any vertex of $\Inducedgraph{G}{1}{b}$ that belongs to $\Level{H}{s+1}{v}$ is also in $\Sep{b}{b+1}$ and hence the minimality of $\OPT{b}{\Sep{b}{b+1}}$ implies that $X_2$ is a minimum cardinality set of type-$2$.
  \end{proof}
  Now we show how to construct a minimum cardinality $s$-CVD set $X_3$ of type-$3$ of $\Inducedgraph{G}{1}{a}-A$. Let $B\subseteq \{1,2,\ldots,a-1\}$ be the set of integers such that for any $i\in B$ the graph 
% $H_i = \Inducedgraph{G}{i+1}{a}-(Q_i\cap Q_{i+1})-A$ 
$H_i = \Inducedgraph{G}{i+1}{a}- (\Sep{i}{i+1} \cup A)$ is connected and
has diameter at most $s$. By definition, a type-$3$ $s$-CVD set $X$ of $H$ contains $S_c^{c+1}$ for some $c \in B$.  We call each such type-$3$ $s$-CVD set as type-$3(c)$. Now we define minimum type-$3(c)$ $s$-CVD set as follows. 
\begin{equation} \label{eq3c}
\begin{split}
\text{For each } c\in B,\quad Z_c & = (\Sep{c}{c+1}-A) \cup \OPT{c}{\Sep{c}{c+1}}\\
 \end{split}
\end{equation}
\begin{myclaim}
The set $Z_c$ is a minimum cardinality $s$-CVD set of type-$3(c)$ of $\Inducedgraph{G}{1}{a}-A$.
\end{myclaim}
\begin{proof}[Proof of Claim]
Note that any vertex of $\Inducedgraph{G}{1}{c}$ that belongs to $A$ is also in $\Sep{c}{c+1}$. By definition, $\Sep{c}{c+1}$ separates the connected component $\Inducedgraph{G}{c+1}{a}- (\Sep{c}{c+1} \cup A)$ from the rest of the graph namely, $\Inducedgraph{G}{1}{c}- (\Sep{c}{c+1})$. Since the diameter of $\Inducedgraph{G}{c+1}{a}- (\Sep{c}{c+1} \cup A)$ is at most $s$ and $\OPT{c}{\Sep{c}{c+1}}$ is the minimal cardinality $s$-CVD set of $\Inducedgraph{G}{1}{c}- \Sep{c}{c+1}$ the set $Z_c = (\Sep{c}{c+1}-A) \cup \OPT{c}{\Sep{c}{c+1}}$ is a minimum cardinality $s$-CVD set of $H$ of type-$3(c)$.
\end{proof}

We define $X_3$ as below.
 \begin{equation} \label{eq3}
 X_3 = \min\{Z_c\colon c\in B\}
\end{equation}
\begin{lemma}
The set $X_3$ is a minimum cardinality $s$-CVD set of type-$3$ of $\Inducedgraph{G}{1}{a}-A$.
\end{lemma}
\begin{proof}
The minimality of each $Z_c$ implies that the set $X_3$ is a minimum cardinality type-$3$ $s$-CVD set.
 \end{proof}
 Finally, we show the construction of a minimum cardinality $s$-CVD set $X_4$ of type-$4$ of $\Inducedgraph{G}{1}{a}-A$.   Let $C\subseteq \{1,2,\ldots,a-1\}$ be the set of integers such that for any $i\in C$ the graph $H_i = \Inducedgraph{G}{i+1}{a}-(\Sep{i}{i+1}\cup A)$ is connected and has diameter exactly $s+1$. By definition, a type-$4$ $s$-CVD set $X$ of $H$ contains $S_i^{i+1}$ for some $i \in C$.  We call each such type-$4$ $s$-CVD set as type-$4(c)$. Now we define minimum type-$4(c)$ $s$-CVD set as follows. Note that $Y_A^a \neq \emptyset$. Let $v$ be some vertex in $Y^a_A$ and $Y_i = \Level{H_i}{s+1}{v}$.  
\begin{equation} \label{eq4c}
\begin{split}
\text{For each } i\in C,\quad Z_i & = (\Sep{i}{i+1}-A) \cup Y_i \cup \OPT{i}{\Sep{i}{i+1}}\\
\end{split}
\end{equation}
\begin{myclaim}
The set $Z_i$ is a minimum cardinality $s$-CVD set of type-$4(c)$ of $\Inducedgraph{G}{1}{a}-A$.
\end{myclaim}
\begin{proof}[Proof of Claim]
Recall that $H_i$ is connected and we claim that the graph $H_i-Y_i$ is also connected: otherwise, if $Y_i$ is a separator of $H_i$ then there exits a vertex $w$ in $H_i-Y_i$ such that $w$ does not belongs to the component containing $v$ in $H_i-Y_i$. Since any path from $v$ to $w$ in $H_i$ passes through $Y_i$, the distance of $w$ from $v$ in $H_i$ is at least $s+2$ contradicting the assumption that $H_i$ has diameter exactly $s+1$.\\
Since $H_i - Y_i$ is connected, it is the \emph{frontal} component of $\Inducedgraph{G}{1}{a}-A-(\Sep{i}{i+1}\cup Y_i)$. Let $A'=A \cup \Sep{i}{i+1}\cup Y_i$. Note that $Y_{A'}^a = Y_{A}^a \neq \emptyset$. Hence the distance between $v \in Y_{A'}^a$ and any other vertex in $H_i-Y_i$ is at most $s$. Thus by Observation \ref{obs:frontal component} the graph $H_i - Y_i$ has diameter at most $s$. Note that $\OPT{i}{\Sep{i}{i+1}}$ is the minimal cardinality $s$-CVD set of $\Inducedgraph{G}{1}{i}- \Sep{i}{i+1}$ and any vertex of $\Inducedgraph{G}{1}{i}- \Sep{i}{i+1}$ that belongs to $Y_i$ or $A$ is also in $\Sep{i}{i+1}$. Hence, the set $Z_i =(\Sep{i}{i+1}-A) \cup Y_i \cup \OPT{i}{\Sep{i}{i+1}}$ is a minimum $s$-CVD set of $H$ of type-$4(c)$.
\end{proof}

Now define $X_4$ as follows.
\begin{equation} \label{eq4}
 X_4 = \min\{Z_i\colon i\in C\}
\end{equation}
\begin{lemma}
The set $X_4$ is a minimum cardinality $s$-CVD set of type-$4$ of $\Inducedgraph{G}{1}{a}-A$.
\end{lemma}
\begin{proof}
The minimality of each $Z_i$ implies that the set $X_4$ is a minimum cardinality type-$4$ $s$-CVD set.
 \end{proof}

Now we define a minimum \textsc{$s$-CVD} set of $\Inducedgraph{G}{1}{a}- A$ as the one with minimum cardinality among the sets $X_i, 1 \leq i\leq 4$. That is,  
\begin{equation}\label{eq:opt}
\OPT{a}{A} = \min\{X_1,X_2,X_3,X_4\}
\end{equation}

A pseudocode of the procedure to find Equation \ref{eq:opt} is given by Procedure \ref{proc:sCD}.

 \begin{minipage}{\linewidth}
 \begin{algorithm}[H]
  \renewcommand*{\algorithmcfname}{Procedure} 
 \DontPrintSemicolon

% % Initialize $B=C=\emptyset$\;
 Let $H=G[1,a]-A$ and $Y_A^a = (Q_a-Q_{a-1})-A$\;
 Set  $X_1 = Y_A^a\cup \OPT{a-1}{A\cap Q_{a-1}}$\;
 For a vertex $v \in Y^a_A$, find the maximum integer $b$ such that $b<a$ and $Q_{b} \cap \Level{H}{s+2}{v} \neq \emptyset$\;
  Set $ X_2 = \Level{H}{s+1}{v} \cup \OPT{b}{\Sep{b}{b+1}} $ \; 
% %  For a vertex $v \in Y^a_A$, find the maximum integer $b'$ such that $b'<i$ and $Q_{b'} \cap \Level{H}{s+1}{v} \neq \emptyset$\;
 Set $B=C=\emptyset$\; 
   \For{$c= 1 \text{ to } a-1$}{
     \If{\emph{Diam}$[c][a][A] \leq s$} {$Z_c=(\Sep{c}{c+1}-A) \cup \OPT{c}{\Sep{c}{c+1}}$\Comment{$H_c=\Inducedgraph{G}{c+1}{a}-(\Sep{c}{c+1}\cup A)$}\;
  $B= B \cup \{c\}$}{
  \If{\emph{Diam}$[c][a][A] = s+1$}{$W_c=(\Sep{c}{c+1}- A) \cup \Level{H_c}{s+1}{v} \cup \OPT{c}{\Sep{c}{c+1}}$\;
  $C= C \cup \{c\}$ }}
  }
% %  Let $B\subseteq \{b',b'+1,\ldots,a-1\}$ be the set of integers such that for any $i\in B$ the graph  $H_i = \Inducedgraph{G}{i+1}{a}-(Q_i\cap Q_{i+1})-A$ has diameter at most $s$
   Set $X_3 = \min\{Z_i\colon i\in B\}$\;
  Set $X_4 = \min\{W_i\colon i\in C\}$\;

   Set $\OPT{a}{A} = \min\{X_1,X_2,X_3,X_4\}$ \;
   Return $\OPT{a}{A}$
   \caption{Compute{\_}sCD$(G,a,A)$}
   \label{proc:sCD}
  
  \end{algorithm}
  \end{minipage}
   We formally summarize the above discussion in the following lemma.
% and for a graph $\Inducedgraph{G}{1}{a}-A$, we call the procedure for finding Equation \ref{eq:opt} by$\text{ \emph{Compute}}{\_}\text{\emph{sCD}}(G,a,A)$.

\begin{lemma} \label{lem:correctness}
For $1 < a \leq k$, if the diameter of the \emph{frontal} component of $\Inducedgraph{G}{1}{a}-A$ is at least $s+1$, then $\OPT{a}{A} = \min\{X_1,X_2,X_3,X_4\}$. 
\end{lemma}
\begin{proof}
The proof follows from Lemma \ref{lem:type} and the above discussion on the minimality of the sets $X_i, 1 \leq i\leq 4$, in their respective types.
\end{proof}
The proof of correctness of the algorithm follows from the Lemmas \ref{lem:recursion-base}, \ref{lem:recursion-1} and \ref{lem:correctness}. A pseudocode of the algorithm for finding a minimum  $s$-CVD set of an interval graph is given in Algorithm \ref{alg:dp}. In the following section, we discuss the time complexity of the algorithm.

% \subsection{Pseudocode of the algorithm}

%  \renewcommand*{\algorithmcfname}{Procedure 1} 
         
%     \caption{$Computing{\_}s-CD{\_}set(H)$} 
  
%     \label{sCD}
\setcounter{algocf}{0}
 \begin{minipage}{\linewidth}
 \begin{algorithm}[H] 
 \DontPrintSemicolon
\begin{small}
   \SetKwFunction{compute}{Compute\_sCD(G,a,A)}
\SetKwInOut{KwIn}{Input}
    \SetKwInOut{KwOut}{Output}
    \KwIn{An interval graph $G$ and a positive integer $s$}
    \KwOut{$\OPT{k}{\emptyset}$}
    %Minimum cardinality subset of vertices, $S\subseteq V(G)$ such that each component of $G[V - S]$ has diameter at most $s$. 
 Using algorithm in \cite{booth1976} find the ordered set of maximal cliques of $G$, say $Q_1,Q_2,\ldots,Q_k$ and $N_{\text{left}}(v)$, $\maxIndex{v}$ and $\minIndex{v}$ for each vertex $v \in V(G)$\;
 
%   Let $t$ be the maximum of $h$ such that the diameter of the induced subgraph $Q_{h,i}$ is exactly $s+1$\;
 
  %Let $\OPT{a}{A}$ stores the minimum $s$-CVD set of the induced subgraph $G[1,a]- A$ where $1 \leq a \leq k$ and $A\in \Allseparate{Q_a}$\;
%   Let \emph{diam}$(G)$ denotes the diameter of the graph $G$\;
 Find $\Allseparate{Q_1}$\;
 %= \{ Q_1 \cap Q_b\},1< b \le k$\;
 \For{all $A \in \Allseparate{Q_1}$}{ $\OPT{1}{A}=\emptyset$}
   
  \For{$a= 2 \text{ to } k$ }{
  Find $\Allseparate{Q_a}$\;
  %= \{ Q_a \cap Q_b\},a< b \le k$\;
   \For{$A \in \Allseparate{Q_a}$}{
Set $Y_A^a = (Q_a - Q_{a-1})-A$\;
 \uIf{$Y_A^a = \emptyset$}{$\OPT{a}{A}=\OPT{a-1}{A\cap Q_{a-1}}$}\Else{
 %Choose a vertex $v \in Y_A^a$\;
%  $N_H(v) = N_{\text{left}} - A$\Comment{Neighbourhood of $v$ in $H=G[1,a]-A$}\;
%  \uIf{$N_H(v)=\emptyset$}{\emph{left\_Index$[a][A]$} =$\epsilon$}\Else{
%  Set $\minIndex{x}=\min\{\minIndex{u'} \colon u'\in N_{H}(v) \}$ \; 
%  Set \emph{left\_Index$[a][A]$} = $\minIndex{x}$\; }
 \For{$c = 1 \text{ to } a-1$}{
 Find the diameter of the induced subgraph $H_c= G[c+1,a]-(A\cup \Sep{c}{c+1})$ using $N_{\text{left}}(v), v \in Y_A^a$ and store it in Diam$[c][a][A]$.  
%  $N_{H_c}(v) = N_{\text{left}} - ( A \cup \Sep{c}{c+1})$\;
%  \Comment{Neighbourhood of $v$ in $H_c=G[1,a]-(A\cup \Sep{c}{c+1})$ \& $H_0=G[1,a]-A$}\;
%  \uIf{$N_{H_c}(v)=\emptyset$ and $c=a-1$}{ \emph{Diam$[c][a][A]$}=1\Comment{Diameter of }}
%   \uElseIf{$N_{H_c}(v)\neq \emptyset$}{
%  Set $\minIndex{y}=\min\{\minIndex{u'} \colon u'\in N_{H_c}(v) \}$\; 
%  \uIf{$\minIndex{y}=a$}{\emph{Diam$[c][a][A]$}=1}
%  \Else{
%  Set \emph{Diam$[c][a][A]$}= Diam$[c][\minIndex{y}][A\cap Q_{\minIndex{y}}] + 1$}}
%  \Else{Set \emph{Diam$[c][a][A]$}=$\infty$}
 }
 \uIf{diameter \emph{Diam$[1][a][A]$} of the \emph{frontal} component of $H_0 = G[1,a]-A \leq s$}{
  $\OPT{a}{A}=\OPT{a-1}{A\cap Q_{a-1}}$}
 \Else{
 $\OPT{a}{A}=\compute$
 %\text{ \emph{Compute}}{\_}\text{\emph{sCD}}(G,a,A)$
 }
 }
 }
 }
 \Return 
   \caption{$s$-CVD$(G,s)$:G is an interval graph and $s$ is a positive integer}
   \label{alg:dp}
    
\end{small}
 \end{algorithm}
\end{minipage}

 \subsection{Time complexity} \label{sec:interval-time}
For a given interval graph $G$ with $n$ vertices and $m$ edges, the algorithm first finds the ordered set of maximal cliques of $G$ as described in Section \ref{sec:interval-notation}. Such an ordered list of the maximal cliques of G can be produced in linear time as a byproduct of the linear (O(n + m)) time recognition algorithm for interval graphs due
to Booth and Leuker~\cite{booth1976}. For each vertex $v\in G$, the algorithm gathers the following information during the enumeration of maximal cliques:  (i) the values $\minIndex{v}$ and $\maxIndex{v}$ and (ii) the set of neighbours of $v$ whose corresponding interval starts before that of $v$ which we call as $N_{\text{left}}(v)$ and are ordered with respect to the left endpoints.\\

Let $Q_1,Q_2,\ldots,Q_k$ be the ordered set of maximal cliques of $G$. From the ordered set of cliques, the algorithm constructs the set $\Allseparate{Q_a}$ (\textbf{steps 2, 6, Algorithm \ref{alg:dp}}) for each $Q_a, 1 \leq a < k$. For an integer $a, 1 \leq a < k$ the set $\Allseparate{Q_a}$ can be constructed by adding a vertex $v \in Q_a$ to each $\Sep{a}{b} \in \Allseparate{Q_a}$ for $a < b \leq \maxIndex{v}$. For the computation of each $\OPT{a}{A}, 1 \leq a \leq k,  A \in \Allseparate{Q_a} $ the algorithm needs to compute the following: (i) the set of vertices, $Y_A^a$ (\textbf{step 8,  Algorithm \ref{alg:dp}}); (ii) the diameter of the frontal component of the graph $H= G[1,a]-A$ (\textbf{step 14,  Algorithm \ref{alg:dp}}) and (iii) the diameter of the induced subgraphs $H_c= G[c+1,a]-(A\cup \Sep{c}{c+1}), 1 \leq c \leq a-1$ (\textbf{steps 12-13, Procedure \ref{proc:sCD}}). 
% Observe that for each maximal clique $Q_b, a < b \leq k$, there exists a vertex $u \in Q_b$ such that $u \not\in Q_{b-1}$: otherwise it contradicts the maximality of $Q_b$. Hence $|\Sep{a}{b}|\leq degree(u)$.  
% % %By definition, $\emptyset \in \Allseparate{Q_a}$.
% Therefore, the time needed for computing each $\Allseparate{Q_a}, 1 \leq a <k$ is atmost $\sum\limits_{u\in Q_b}degree(u)$.  Distinct entries in $\Allseparate{Q_a}$ can be find out from the values of $\maxIndex{v}, v \in Q_a$/
% Since $G$ has at most $n$ maximal cliques and $\sum\limits_{u\in Q_b}degree(u) = O(m)$, the total time needed for computing all $\Allseparate{Q_a}, 1 \leq a \leq k$ is at most $O(n\cdot m)$. 
% Remaining steps in the algorithm iterates only for distinct values in $\Allseparate{Q_a}$ which can be find out from the values of $\maxIndex{v}, v \in Q_a$ because for $A_1, A_2 \in \Allseparate{a}$,where $A_1 =\Sep{a}{b}$ and $A_2 =\Sep{a}{b'}$.

% For an iterative construction of the table $\Psi$, in each iteration the algorithm solves the set of distinct subproblems $\OPT{a}{A}, A \in \Allseparate{Q_a}$ corresponding to each maximal clique $Q_a , 1 \leq a \leq k$. %That is, $A =\Sep{a}{b}$ only if $\Sep{a}{b} \neq \Sep{a}{b'}, b'\neq b$. 

The set $Y_A^a$ can be obtained from the vertex set of $Q_a$ in linear time by checking the $\minIndex{v}$ and $\maxIndex{v}$ values of each vertex $v \in Q_a$. That is, $Y_A^a = \{v \in Q_a: \minIndex{v}=a \text{ and } \maxIndex{v} < b, A = \Sep{a}{b}\}$. Let Diam$[1][a][A]$ be the  diameter of the frontal component of $H= G[1,a]-A$. By Observation \ref{obs:frontal component}, diameter of the frontal component of $H$ is equal to the eccentricity of a vertex $v \in Y_a^A$. That is, the maximum distance of $v$ from other vertices in $H$ which we denote by \textit{ecc$_{H}(v)$}. Hence, Diam$[1][a][A]=$\textit{ecc$_{H}(v)$}. Let $v_l$ be the \emph{leftmost} neighbour of $v$ in $H$ such that $\minIndex{v_l}= a'$ and \textit{ecc$_{H'}(v_l)$} be the eccentricity of $v_l$ in $H' = G[1,a']- (Q_{a'} \cap Q_b)$. Then observe that \textit{ecc$_{H}(v)$}= \textit{ecc$_{H'}(v_l) +1$}. Therefore, Diam$[1][a][A]$ = Diam$[1][a'][Q_{a'} \cap Q_b] + 1$. Since the \emph{leftmost} neighbour of $v$ in $H$ can be found in linear time from $N_{\text{left}}(v)$ by checking the $\minIndex{v}$ and $\maxIndex{v}$ values of each vertex $u \in N_{\text{left}}(v)$, diameter of the frontal component of $H$ can be found in $O(n)$ time. Similarly, diameter of the induced subgraphs $H_c= G[c+1,a]-(A\cup \Sep{c}{c+1})$ in \textbf{steps 12 -13, Procedure \ref{proc:sCD}} together can be found in $O(n)$ time by similar arguments as above and the following observation; $N_{\text{left}}(v) - (A \cup \Sep{c}{c+1}) \supseteq N_{\text{left}}(v) - (A \cup \Sep{c+1}{c+2})$.
% \begin{enumerate}[label=\roman*., itemsep=0pt, topsep=0pt]
%     \item $Y_A^a$
%     \item  $N_H(v)$
%     \item $N_{H_c}(v)$
%     \item \emph{Diam$[c][a][A]$}
    
% \end{enumerate}
To compute the overall time complexity of our algorithm, we have the following claims.
\begin{claim}
Total number of subproblems computed by the algorithm, Algorithm~\ref{alg:dp}
%\textbf{steps 7-14}
is at most $O(|V|+ |E|)= O(n+m)$.
 \end{claim}
\begin{proof}[Proof of Claim]
Note that with respect to the ordering of maximal cliques of $G$ the elements of the set $\Allseparate{Q_a}$ have the following relation. For each $b, a < b \leq k$ we have $S_a^{b+1} \subseteq S_a^b$. Hence the number of distinct subproblems computed by the algorithm corresponding to each maximal clique $Q_a$ is at most $|S_a^{a+1}|+1$ (Recall that one of the subproblem corresponds to $\emptyset \in \Allseparate{Q_a}$). Since the number of maximal cliques in $G$ is at most $|V|=n$ and $|S_a^{a+1}| \leq degree(v), v \in Q_a -Q_{a+1} $, the total number of subproblems computed by the algorithm
%\textbf{steps 7-14}
is at most $\sum\limits_{v \in Q_a -Q_{a+1} }degree(v) + |V| \le O(|V|+ |E|)= O(n+m)$.
\end{proof}
\begin{claim}
The procedure \emph{ Compute{\_}sCD$(G,a,A)$} computes the minimum cardinality $s$-CVD set of $H=\Inducedgraph{G}{1}{a}- A$ in $O(n)$ time.
 \end{claim}
 \begin{proof}[Proof of Claim:]
 Observe that the time complexity of the procedure \emph{Compute{\_}sCD$(G,a,A)$} depends mainly on building the sets $X_i, 1 \leq i \leq 4$. Since the set $Y_A^a, 1 \leq a < k, A \in \Allseparate{Q_a}$  is obtained in $O(n)$ time, the set $X_1$ can be computed in $O(n)$ time.\\
 The set $\Level{H}{s+1}{v}$ can be computed from the leftmost neighbour of $v$ in $H$, say $v_l$ in linear time by $s$ iterations: In the first iteration, find the leftmost vertex of $v_l$ in $N_{\text{left}}(v_l)-A$, in the second iteration find the leftmost vertex in the second neighbourhood  and so on. Moreover, the leftmost neighbour of $v$ in $H$ can be obtained by a linear search of $N_{\text{left}}(v)$.
%  Let $x$ be a neighbour of $v$ such that $$\minIndex{x} = \min\{\minIndex{u'} \colon u'\in N_H(v) \}$$.
%  Let $H'=\Inducedgraph{G}{1}{\minIndex{x}}- (A\cap Q_{\minIndex{x}})$. Observe that the set $\Level{H}{s+1}{v}$ is exactly the set $\Level{H'}{s}{x}$ which we already computed.  Now consider the induced subgraph $H_{i}=\Inducedgraph{G}{i+1}{a}- (\Sep{i}{i+1}\cup A), 1 \leq i < a$. Let $y$ be a neighbour of $v$ such that $$\minIndex{y} = \min\{\minIndex{u'} \colon u'\in N_{H_i}(v) \}$$. Observe that the diameter of $H_i$ can be obtained by adding $1$ to the diameter of $\Inducedgraph{G}{i+1}{\minIndex{y}}- (\Sep{i}{i+1}\cup (A\cap Q_{\minIndex{y}}))$ unless $\minIndex{y} = a$: then the diameter of $H_i$ is $1$. 
 Since the number of induced subgraphs $H_c$ is at most $O(n)$, the sets $X_3$ and $X_4$ can be constructed in $O(n)$ time. Hence the claim follows.
 \end{proof}
Therefore, by the above claims the overall time complexity of our algorithm is $O(n\cdot(n+m))$ and Theorem \ref{thm:interval-scd} follows.

\section{Hardness for well-partitioned chordal graphs}\label{sec:hardness}.

In this section, we prove Theorem~\ref{thm:hardness}. We shall use the following observation.

\begin{observation}\label{obs:build-well-partitioned}
Let $H$ be a well-partitioned chordal graph. Let $H'$ be a graph obtained from $H$ by adding a vertex of degree $1$. Then $H'$ is an well-partitioned chordal graph.
\end{observation}

\newcommand{\Gwell}[1]{#1_{well}}
Let $s\geq 2$ be an even integer and let $s=2k$. We shall reduce \textsc{Minimum Vertex Cover (MVC)} on general graphs to \textsc{$s$-CVD} on well partitioned graphs. Let $\langle G,k \rangle$ be an instance of \textsc{Minimum Vertex Cover} such that maximum degree of $G$ is at most $n-3$. Let $\overbar{G}$ denote the complement of $G$. Now construct a split graph $\Gwell{G}$ from $G$ as follows. For each vertex of $v\in V(G)$, we introduce a new path $P_v$ with $k-1$ edges and let $x_v,x'_v$ be the endpoints of $P_v$. For each edge $e\in E\left(\overbar{G}\right)$ we introduce a new vertex $y_e$ in $\Gwell{G}$. For each pair of edges $e_1,e_2 \in E(\overbar{G})$ we introduce an edge between $y_{e_1}$ and $y_{e_2}$  in $\Gwell{G}$. For each edge $e=uv \in E\left(\overbar{G}\right)$, we introduce the edges $x_u y_e$ and $x_v y_e$ in $\Gwell{G}$. 
Observe that $C=\{y_e\}_{e\in E\left(\overbar{G}\right)}$ is a clique, $I=\{x_v\}_{v\in V(G)}$ is an independent set of $\Gwell{G}$. Therefore $C\cup I$ induces a split graph, say $G'$, in $\Gwell{G}$. Since $\Gwell{G}$ can be obtained from $G'$ by adding vertices of degree $1$, due to Observation~\ref{obs:build-well-partitioned}, we have that $\Gwell{G}$ is an well-partitioned graph. We shall show that $G$ has a vertex cover of size $k$ if and only if $\Gwell{G}$ has a \textsc{$s$-CVD} set of size $k$. 

% \begin{figure}
%     \centering
%     \includegraphics[scale=0.5]{split.pdf}
%     \caption{Reduction for Theorem~\ref{thm:split-hard}. (a) The graph $G$, (b) The graph $\overline{G}$, and (c) The graph $\Gwell{G}$.}
%     \label{fig:split}
% \end{figure}

\begin{observation}\label{obs:neighbour}
For each vertex $v\in C$, $|N[v]\cap I|=2$ and for each vertex $u\in I$, $|N[u]\cap C|\geq 2$.
\end{observation}

\begin{lemma}\label{lem:vertex-cover-2-CVD}
Let $D$ be a subset of $I$ and let $T=\{u\in V(G)\colon x_u \in D\}$. The set $D$ is a \textsc{$s$-CVD} set of $\Gwell{G}$ if and only if $T$ is a vertex cover of $G$.  
\end{lemma}

\begin{proof}

Let $D'=\{x'_v\colon x_v \in I-D\}$ and $T'=\{u\in V(G)\colon x_u \in D'\}$ (note that $T=V(G) - T'$). Note that there is one single component $G'$ of $\Gwell{G} - D$ that contains vertices from $C$  since there are no isolated vertices by observation \ref{obs:neighbour}. Observe that $G'$ contains $I-D$. Therefore, for any two vertices $x'_u,x'_v\in D'$ the distance between $x'_u,x'_v$ is $s$ if and only if there is an edge between $u,v$ in $\overbar{G}$. Therefore, distance between any two pair of vertices in $D'$ is $s$ if and only if $T'$ induces a clique in $\overbar{G}$ and therefore an independent set in $G$. Since $T=V(G)- T'$, we have that distance between any two pair of vertices in $D'$ is $s$ if and only if $T$ is a vertex cover of $G$. Since $|D'|=|I-D|$ we have that $D$ is an $s$-CVD set of $\Gwell{G}$ if and only if $T$ is a vertex cover of $G$.
\end{proof}

\begin{lemma}\label{lem:independent-set}
There is a subset of $I$ which is a minimum \textsc{$s$-CVD} set of $\Gwell{G}$.
\end{lemma}
\begin{proof}
Let $S$ be a minimum \textsc{$s$-CVD} set of $\Gwell{G}$ such that $|S \cap I|$ is maximum. We claim that $S \subseteq I$. Suppose for contradiction this is not true. Let $I'=\bigcup\limits_{u\in V(G)} P_u-\{x_u\}$. Then we must have that $S\cap I' \neq \emptyset$ or $S \cap C \neq \emptyset$. Let $a$ be a vertex of $S\cap I'$. Observe that there must be a vertex $u\in V(G)$ such that $a\in P_u$ and that $(S-\{a\})\cup \{x_u\}$ is an $s$-CVD set of $\Gwell{G}$. This contradicts the assumption that $S$ is a minimum $s$-CVD set of $\Gwell{G}$ with $|S\cap I|$ maximum. 

Now consider the collection $\mathcal{C}$ of connected components of $\Gwell{G}-S$. First, observe that there exists at most one connected component in $\mathcal{C}$ that intersects $C$ (the clique of $\Gwell{G}$). We shall call such a component as the \emph{big component} and let $X$ be the set of vertices of the big component. In fact $I$ itself is a \textsc{$s$-CVD} set and observation \ref{obs:neighbour} implies $|I| \le |C|$. Therefore, without loss of generality we can assume that $C \not\subset S$ and indeed such a big component exists. 

Let $Y$ denote those vertices of $\Gwell{G}-S$ that belongs to $I - X$. Let $S_C = S\cap C$ and $S_I = S \cap I$. Recall that by assumption, $S_C \neq \emptyset$.

If there is a vertex $v \in S_C$ such that $|N[v]\cap Y| = 0$, then $S - \{v\}$ is a \textsc{$s$-CVD} set with $X \cup \{v\}$ as corresponding big component with diameter less than or equal to $s$. This contradicts the minimality of $S$. Similarly, if there exists a vertex $v \in S_C$ such that $N[v]\cap Y = \{u\}$, a singleton set then $S' = S \cup \{u\} - \{v\}$ is a new \textsc{$s$-CVD} set with $X \cup \{v\}$ as corresponding new big component. This contradicts the assumption that $S$ is a minimum \textsc{$s$-CVD} set  with $|S\cap I|$ is maximum. Hence together with observation \ref{obs:neighbour} we infer that $|N(v)\cap Y|=2$, for each $v\in S_C$. Observation \ref{obs:neighbour} also implies that for each vertex $u \in Y$, $|N(u)\cap S_C|\ge2$, since $Y \subseteq I$ for each $u \in Y$ we have $N(u)\subseteq S_C$. Therefore, $|Y| \le |S_C|$ and $S'=(S- S_C)\cup Y$ is a minimum \textsc{$2$-CVD} set with $X \cup S_C$ as the corresponding new big component and $|S' \cap I| > |S \cap I| $. This contradicts the assumption for $S$.

Hence we conclude that $S$ is indeed a minimum \textsc{$s$-CVD} set such that $S \subseteq I$.
\end{proof}

Lemmas~\ref{lem:vertex-cover-2-CVD} and~\ref{lem:independent-set} imply that $G$ has a vertex cover of size $k$ if and only if $\Gwell{G}$ has a \textsc{$s$-CVD} set of size $k$. Now Theorem~\ref{thm:hardness} follows from a result of Khot and Regev~\cite{khot2008}, where they showed that unless the Unique Games Conjecture is false, there is no $(2-\epsilon)$-approximation algorithm for \textsc{Minimum Vertex Cover} on general graphs, for any $\epsilon>0$.
\section{Conclusion}\label{sec:conclude}

In this paper we studied the computational complexity of \textsc{$s$-CVD} on well-partitioned chordal graphs, a subclass of chordal graphs which generalizes split graphs. We gave a polynomial-time algorithm for $s=1$ and we proved that for any even integer $s\geq2$, \textsc{$s$-CVD} is NP-hard on well-partitioned chordal graphs. We also provide a faster algorithm for \textsc{$s$-CVD} on interval graphs for each $s\geq 1$. This raises the following questions.
\begin{question}
    What is the time complexity of Cluster Vertex Deletion on chordal graphs?
\end{question}

\begin{question}
    What is the time complexity of \textsc{$s$-CVD}  on chordal graphs for odd values of $s$?
\end{question}

\begin{question}
    Is there a constant factor approximation algorithm for $s$-CVD, $s\geq 2$ on chordal graphs?
\end{question}

Another generalisation of interval graphs is the class of \emph{cocomparability} graphs. It would be interesting to investigate the following question.

\begin{question}
    What is the time complexity of \textsc{$s$-CVD}  on cocomparability graphs for each $s\geq 1$?
\end{question}

\end{document}